\documentclass[11pt]{article}
\usepackage{macros/packages}
\usepackage{macros/statistics-macros}
\usepackage{macros/editing-macros}
\usepackage{macros/formatting}

\begin{document}

\begin{center}
  {\LARGE Modeling Interference Using Experiment Roll-out} \\
  \vspace{.4cm}
  {\large Ariel Boyarsky$^1$ ~~~~  Hongseok Namkoong$^1$ ~~~~ Jean Pouget-Abadie$^2$} \\
  \vspace{.2cm}
  {$^1$Decision, Risk, and Operations Division, Columbia Business School, $^2$Google Research}  \\
  \vspace{.2cm}
  \texttt{\{aboyarsky26, namkoong\}@gsb.columbia.edu, jeanpa@google.com}
\end{center}

\vspace{.2cm}

\begin{abstract}
Experiments on online marketplaces and social networks suffer from \emph{interference}, where the outcome of a unit is impacted by the treatment status of other units. We propose a framework for modeling interference using a ubiquitous deployment mechanism for experiments, staggered  \emph{roll-out} designs, which slowly increase the fraction of units exposed to the treatment to mitigate any unanticipated adverse side effects. Our main idea is to leverage the temporal variations in treatment assignments introduced by roll-outs to model the interference structure. Since there are often multiple competing models of interference in practice we first develop a  model selection method that evaluates models based on their ability to explain outcome variation observed along the roll-out. Through simulations, we show that our heuristic model selection method, \textsc{Leave-One-Period-Out}, outperforms other baselines. Next, we present a set of model identification conditions under which the estimation of common estimands is possible and show how these conditions are aided by roll-out designs. We conclude with a set of considerations, robustness checks, and potential limitations for practitioners wishing to use our framework.
\end{abstract}

% \newpage 
% \tableofcontents
% \newpage 

\section{Introduction}
\label{sec:intro}

Experimentation is at the core of scientific decision-making in many online platforms. However, many of the experiments run by these platforms suffer from interference, meaning an intervention on a participant might impact the outcome of other participants. For example, interference is problematic in online marketplaces where strategic agents compete for scarce resources~\citep{pouget2018optimizing, Ugander2020graph, li2021twosided, Niazdeh2021cluster, brennan2022cluster, Bojinov2022switchbacks, Bojinov2022Online, lobel2022interference}, and in online social networks~\citep{eckles2016estimating, Ugander2022Spillover} where peers interact and influence each other. Interference among units leads to a violation of the Stable Unit Treatment Value Assumption~\citep{Rubin2005POs,imbens2015causal} and the causal effect may not even be identifiable in a randomized experiment.

A common practical approach to estimating causal effects when interference is present is to assume a specific potential outcome structure based on the network(s) between participants. Together, the postulated potential outcome structure and the underlying network(s) constitute an \emph{interference model}, which can be used to estimate the true causal effect directly~\citep{eckles2017design,biswas2018estimating}, or modify the experimental design to mitigate the effects of interference~\citep{basse2018model}. 
As a motivating example, which we will explore further in Section~\ref{sec:setting}, consider an online auction platform incrementally rolling out an experiment that tests the effectiveness of increasing reserve prices. Treating one auction is likely to affect each participating advertisers, which can in turn affect other auctions in which they participate. Practitioners have postulated several working models of interference for how advertisers and auctions interact with one another.  Each modeling choice might lead to a different estimated treatment effect. %By examining the temporal variations of treatment exposure induced by the roll-out design, we show that practitioners can answer each of these questions about the presence of interference, the potential outcomes structure, and the structure of the interference network.

Such ad-hoc structural modeling, however, is often unreliable in practice since the interference network is rarely known and the potential outcome structure is almost always misspecified.  As a result, selecting and validating models of interference is a crucial---but difficult---step of estimating causal effects when interference is present. Furthermore, even when committing to a fixed interference model, estimating its parameters may be infeasible if the observed data does not contain enough variation in treatment exposure across units.

We propose a framework for modeling interference using a ubiquitous deployment mechanism for experiments, \emph{staggered roll-out designs}. Online platforms rely on staggered roll-out designs~\citep{kohavi2009controlled, xu2015infrastructure, xu2018sqr} as an early detection tool for any unintended consequences caused by a new experiment or product launch, e.g., software bugs or adverse participant responses. Roll-out designs follow a simple principle: instead of intervening on all participants marked for the intervention all at once, the proportion of participants exposed to the intervention is increased incrementally until all participants marked for the intervention have been intervened on. The practice is nearly universal across standard experimentation infrastructures; for example,~\citet{xu2018sqr} notes that out of 5000+ experiments run annually at LinkedIn, ``every experiment goes through a \emph{ramp-up/roll-out} process.''

In this paper, we propose to leverage existing roll-out designs for a new purpose: to better select and estimate models of interference in causal estimation when interference is present. Despite its ubiquity, the temporal variation in treatment exposure 
induced by roll-out designs is a promising yet largely overlooked consequence of how experiments are implemented in practice. Variations in treatment proportions, whether temporal~\citep{bojinov2022design, Li2022} or spatial~\citep{athey2018exact, baird2018optimal}, are key to modeling interference and validating modeling choices.  

% We propose to leverage  from a causal inference
% perspective.  Rolling out an experiment induces temporal variation in
% treatment exposures, allowing the modeler to observe potential outcomes with
% multiple interference patterns. The temporal variation can crucially
% differentiate between multiple models of interference, as a good model must be
% able to explain how outcomes shift over varying levels of interference induced
% at different treatment exposure levels. (In the following, a \emph{model of
%   interference} refers to a set of postulated interference networks and a
% data-generating model specifying how interference networks impact the
% potential outcomes.)  

As our first main contribution, we utilize roll-out experimental designs to develop a model selection mechanism, \textsc{Leave-One-Period-Out}, to select the best interference model and evaluate this mechanism over a rich set of simulated examples. We show that, given an
interference model,  roll-out designs allow the identification of interference parameters, which may be unidentified in the absence of any temporal variation in treatment exposure. We theoretically characterize when the causal estimand becomes identifiable with the help of roll-out designs and quantify the level of temporal variation required to identify heterogeneous patterns across units. Finally, we quantify the statistical efficiency gains resulting from a roll-out design. 
%The ability to validate and select between interference models allows flexible
%and engineered hypotheses. 
%Our models can utilize multiple interference networks (where edges represent how units/nodes interfere with each other) constructed from engineered features, as well as multiple potential outcome structures.

% Our framework allows flexible design of interference
% models, roll-out-based model selection, and reliable identification of
% heterogeneous causal effects.

In Section~\ref{sec:setting}, we introduce our estimand of interest, define roll-out designs, and provide a motivating example.  In Section~\ref{sec:Model Selection}, we develop our heuristic modeling framework and provide experimental evidence supporting our model selection algorithm. In Section~\ref{sec:Identification}, we study the question of identification, specifically how roll-outs  help identify and estimate causal effects in the presence of interference. Finally, in Section~\ref{sec:considerations}, we address challenges practitioners may face when applying our framework and consider possible robustness checks that could be performed.

% \jpa{is this concluding paragraph necessary or is it repeating what's already been said before}
% Our operational perspective leverages existing experimentation infrastructure
% for roll-out designs and allows reliable causal estimation under interference
% through flexibly designed models of interference. Our approach is inspired by
% the common aphorism “all models are wrong, but some are useful”: the proposed
% model selection method picks out ``useful'' models of interference. This is in
% contrast to previous approaches (e.g.,~\citep{basse2016randomization,
%   li2021twosided, yu2022experiment, liwager2022randomgraph} that require
% committing to a single model of interference to allow statistical inference.

% This is in
% contrast to previous approaches (e.g.,~\citep{basse2016randomization,
%   li2021twosided, yu2022experiment, liwager2022randomgraph} that require
% committing to a single model of interference obeying strong but unverifiable
% assumptions on the interference network in order to allow statistical inference. \jpa{this last sentence feels a bit adversarial.}

\subsection*{Related work}
\label{sec:related_works}

The literature on causal inference in the presence of interference is extensive (e.g., see the work of \citet{Rosenbaum2007Interference,Hudgens2008Interference, Savje2017ATE, Leung2019Interference, Farias2022Markov, viviano2020interference}). Our work is related to the subset of this literature that focuses on new experimental designs to mitigate the effects of interference~\citep{eckles2017design, baird2018optimal, brennan2022cluster}; we repurpose roll-outs, a common design in practice,  to model interference.  
While several authors have proposed new estimands and estimators to better understand the mechanism of interference and reduce its bias-inducing effects~\citep{yuan2021causal, karrer2021network, yu2022estimating, zigler2021bipartite}, we study common models of interference~\citep{aronow2017estimating, basse2016randomization} and focus on the familiar total treatment effect estimand~\citep{chin2019regression}. Instead of suggesting new estimators and estimands, we leverage roll-out designs to formulate a new method for validating and choosing from these previously introduced models.

The operational benefits of roll-out designs have been well-documented in the context of online platforms~\citep{kohavi2009controlled, xu2015infrastructure, xiong2020optimal}. However, the study of roll-out designs is sparse in the context of interference. \citet{yuetal2022rollout} recently studied roll-out designs to develop unbiased estimates of treatment effects. The work is similar to ours in that both frameworks do not assume the knowledge of the underlying network but make different assumptions to make inference possible. 
\citet{yuetal2022rollout} takes a design-based perspective in which given a low-degree polynomial structure for interference they design a roll-out with an unbiased estimator of the treatment effect. In contrast, our heuristic model selection framework applies to any roll-out design and interference model. When the interference model is well-specified, we provide theoretical guarantees quantifying how roll-out designs boost statistical efficiency.

Validating and choosing from potential outcome models that incorporate interference bears many similarities to the task of detecting interference, which has historically relied on one of two methods. The first set of methods is to compare two designs with different properties under SUTVA and interference, often simultaneously using a hierarchical design structure~\citep{sinclair2012detecting, saveski2017detecting, pougetabadie2019testing}. A second approach consists in running Fisher-randomized-like tests on observed data to determine the significance of well-chosen estimators that are non-zero if and only if interference is present~\citep{aronow2017randomizationinference, athey2018exact, basse2019randomization}. Both approaches seek to exploit fluctuations in a specific parameter of interference to determine whether (1) interference is present and potentially (2) whether it occurs in the form or through the channel of that parameter. 

Roll-outs provide a third paradigm for detecting interference, as they introduce desirable fluctuations in interference-sensitive parameters such as the global treatment fraction. In a concurrent and independent work,~\citet{Li2022} study how roll-outs can be used to detect interference. They design randomization tests that can be used to detect cross-unit interference even in the presence of temporal effects. In comparison, the present paper goes beyond detection: at the cost of stronger modeling assumptions, we study the direct modeling of interference effects. Since interference may be common in online platforms, rather than establishing its existence, we propose estimation and model selection  methods that allow contextualizing the operational significance of interference effects compared to the direct treatment effect. While our modeling approach can also be used as a heuristic test for interference, we do not establish its theoretical validity and instead focus on identifiability and estimation guarantees.

% ~\citet{xiong2020optimal} considers the optimal design of adaptive and non-adaptive roll-out experiments to maximize the precision of causal estimates. When interference is present, \citet{xiong2020optimal} advocate for aggregating to levels at which interference is no longer a concern. This is fairly different from our work which focuses on leveraging temporal variation generated by roll-outs to model cross-unit interference and estimate causal effects. While we also consider variance reductions due to roll-outs, our focus is not on maximizing the variance reduction but rather leveraging statistical efficiency to identify causal effects in a finite population setting. 

\section{Setting}
\label{sec:setting}

We use the potential outcomes notation for a finite
population of size $N$. We do not make the Stable Unit Treatment Value Assumption (SUTVA)~\citep{imbens2015causal} such that, for any treatment vector $z \in \{0,1\}^N$, the potential outcome $Y_i(z)$ of each unit $i$ may depend on the treatment status of other units due to interference. 
For concreteness, we focus on the identification and estimation of the total treatment effect estimand
\begin{equation}
\label{eq:TTE}
\TTE \defeq \frac{1}{N} \sum_{i=1}^N 
\left( Y_i(\mymathbb{1}) - Y_i(\mymathbb{0}) \right).
\end{equation} 
In practice, the TTE is of particular relevance to online platforms whose goal is to determine whether a product innovation is fruitful when it is completely adopted~\citep{Bond2012voter,eckles2017design}.

Roll-outs, also known as ``ramp-ups'', are a common experimental practice where instead of assigning treatments at once, it is done in incremental steps. While primarily instituted to improve engineering reliability, they induce important temporal variation in a unit's  treatment exposure that can be used to better model the effect of interference in randomized experiments.  Formally, a roll-out design with $T$ periods consists of a  sequence of treatment assignments $\{Z^t\}_{t=1}^T$ and corresponding observed
outcomes $\{Y_{i, t} \defeq Y_i(Z^t)\}_{t \in [T], i \in [N]}$ such that, once
a unit is treated, it remains treated for the remainder of the experiment. For
typical experiments, $T$ is a small number, often equal to $5$ or less. 
%We will use the notation $Z_i^t$ to refer to unit $i$'s treatment status after period $t$. 
The \textit{completely randomized roll-out design} considers a fixed proportion of units to be newly treated at each period~\citep{Lee2022}. 

%\hn{I don't think the joint distribution of $Z^t$ is specified with this definition. There's no mention of how $Z_i^t$ is actually generated (completely at random).}

\begin{definition}[Completely Randomized Roll-outs] \label{def:CRD} A \textbf{$T$-period completely randomized roll-out} is an increasing set of random treatment assignments, $\{Z^1, \dots, Z^T\}$, and treatment allocation vector $\vec{p}=\{p_1,\dots,p_T\}$ with $\sum_{t=1}^T p_t \leq 1$ such that in each period, $t$, $Z_i^t \in \{0,1\}$ is randomly chosen such that $\sum_{i=1}^N Z_{i}^t=\left\lfloor N\sum_{j=1}^t p_t \right\rfloor$ and if $Z_{i}^{t-1}=1$ then $Z_{i}^t=1$. 
\end{definition}
\noindent Another roll-out design is to specify an independent Bernoulli probability to treat each unit, known as a \textit{Bernoulli randomized design}. This design does not meaningfully change the conclusions of our work, and we defer its definition to Section \ref{subsec:rollout} of the appendix. 
Roll-out designs are characterized by the proportion of newly treated individuals in each period: ``even'' (resp. ``uneven'') roll-outs treat the same (resp. different) incremental proportions of individuals in each period. Even among even and uneven roll-outs, there are several possible roll-out mechanisms for assigning treatments, corresponding to different joint distributions over $Z^t$. 

% This is also the motivating idea of the randomization tests in~\cite{Li2022}, which design tests for interference by leveraging this variation. As we show later, roll-outs are useful beyond tests for interference: they can also be used to choose and fit models of interference.

\vspace{10pt}
\begin{example}[Linear-in-Means Models with Heterogeneity]
  \label{ex:linearmodels}
  As a motivating example, we consider an advertising auction system where bidders compete for limited items. We are interested in measuring the impact of changing the reserve price---the minimum required bid to participate in the auction---on advertisers' spend (outcomes). The example models common operational concerns on online platforms~\citep{pouget2018optimizing}.  In this two-sided marketplace with finite resources, interference occurs when changing the reserve price for some items leads bidders to change their bidding strategy, thus affecting the outcome of other auctions they participate in. 
  
  While a bipartite graph of bidders and auctions is usually used to represent the full market, in statistical inference, it is common to reduce the bipartite market structure to a single interference graph between the $N$ items~\citep{brennan2022cluster}. In this interference graph, edges represent a notion of competition, e.g., substitutable keywords (goods). As we substantiate further in Section~\ref{sec:Model Selection}, there are multiple ways to construct the item-to-item interference network in practice: we may consider whether differences in advertising budgets should be taken into account or whether two items are considered in competition if their co-bidders achieve a certain activity threshold.
  
  For concreteness, consider two plausible and competing ways of defining
  ``neighboring units'' for any given unit $i$, $\mathcal{G}_1(i)$ and $\mathcal{G}_2(i)$,
  based on two different interference
  networks. For each notion of ``neighborhood'', we can posit a simple linear model of
  interference
  \begin{subequations}
  \label{eqn:two-models}
  \begin{align}
    Y_i^t(Z^t) & = \alpha_i\opt    + \tau\opt   \cdot Z_{i}^t
                        + \eta_1\opt   \cdot \sum_{j \in \mathcal{G}_1(i)} Z_{j}^t + \epsilon_{i,t} 
                            \label{eqn:N1-model} \\
    Y_i^t(Z^t) & =  \alpha_i\opt    + \tau\opt   \cdot Z_{i}^t + \eta_2\opt     \cdot \sum_{j \in \mathcal{G}_2(i)} Z_{j}^t + \epsilon_{i,t}.
    \label{eqn:N2-model}
  \end{align}
  \end{subequations}
   Linear models similar to the ones above have been previously studied by~\citet{eckles2017design, aronow2017estimating, basse2019formation}. Given these two competing models of interference~\eqref{eqn:two-models}, the better model can be selected by measuring each model's ability to explain the variation in the outcomes as treatment exposure increases in the roll-out. 
\end{example}

\section{Model Selection}
\label{sec:Model Selection}

To motivate the model selection problem, we again consider the class of models defined in Example \ref{ex:linearmodels}. This example captures how a standard linear model can be highly flexible, allowing us to incorporate a wide range of interference structures. Its richness highlights the need to distinguish between model instances that are useful in explaining interference and those that are not. 

In this section, we propose a model selection mechanism inspired by leave-one-out cross-validation to choose between models of interference.  A unit's outcome depends on its treatment exposure which varies as we increase the treatment allocation throughout a roll-out. Our main observation is that we can test whether the selected interference model, trained on a subset of treatment periods, is able to extrapolate to different levels of treatment exposure by evaluating its predictive performance on observations from remaining periods. Hence, we use the mean-squared prediction error for a given period $t$,
\begin{equation}\label{eq:MSPE}
    \MSPE_t(\what{\theta}) \defeq \frac{1}{N}\sum_{i=1}^{N} (X_i^t\what{\theta}-Y_i^t)^2 = \frac{1}{N}\sum_{i=1}^{N} (\hat{Y}_i^t(\what{\theta})-Y_i^t)^2 ,
\end{equation}
where $Y_i^t(\what{\theta})=X_i^t\what{\theta}$, $X_i^t \in \R^{1\times K}$ refers to the row of features for unit $i$ at time $t$, and $\what{\theta}\in\R^K$ are estimated parameters. We can relate the $\MSPE$ to the estimation error of the $\TTE$. For instance, suppose that we are given data for a new roll-out period, $s$, of the form $({X}^{s}, {Y}^{s})$ where ${X}^{s} \in \R^{N\times K}$ and ${Y}^{s} \in \R^{N}$ (cf. Eq. \eqref{eq:features_period}). Using the previous roll-out periods to estimate $\what{\theta}$, we predict $\hat Y^{s}(\what{\theta}) = X^{s}\what{\theta}$. Define the sample covariance matrix as $\Sigma^{(s)} = \frac{1}{N}\sum_{i=1}^N {X}_{i}^{s}{X}_{i}^{s\top}$. Then, 
\begin{align*}
(\hat\TTE - \TTE)^2  &= [c^\top(\theta\opt    - \what{\theta})]^2
\\&\leq \ltwo{c}^2 \cdot \frac{||\what{\theta} - \theta\opt||^2_{\Sigma^{(s)} }}{\lambda_{\min}(\Sigma^{(s)} )} 
= \ltwo{c}^2 \frac{\frac{1}{N}\sum_{i=1}^N(X^{s}_i\what{\theta} - Y_i^{s}+\epsilon_{i,s})^2}{\lambda_{\min}(\Sigma^{(s)} )} 
\\&\leq  \frac{2\ltwo{c}^2}{\lambda_{\min}(\Sigma^{(s)} )}\cdot\frac{1}{N}\sum_{i=1}^N\left[(X^{s}_i\what{\theta} - Y_i^{s})^2+\epsilon_{i,s}^2\right]
= \frac{2\ltwo{c}^2}{\lambda_{\min}(\Sigma^{(s)} )}\cdot\left(\MSPE_{s}(\what{\theta}) + \frac{1}{N}\sum_{i=1}^N\epsilon_{i,s}^2\right)
\end{align*}
 where $\epsilon_{i,s} = Y_i^{s} - X_i^{s}\theta\opt   $ and the second inequality follows from convexity. 
 % While we cannot directly control $\epsilon_{i,s}$ without further assumptions (e.g. under Assumptions \ref{assum:fixed} or \ref{assum:iid} it is possible to further bound this quantity in high probability using Hoeffding's inequality) 
 % Although $\lambda_{\min}(\Sigma^{(s)})$ may still depend on $\what{\theta}$, 
 This provides a heuristic argument for using the $\MSPE$ criteria in our model section procedure.

\paragraph{Proposed method: Leave-One-Period-Out}

\begin{figure}[t]
    \centering
    \includegraphics[width=0.60\textwidth]{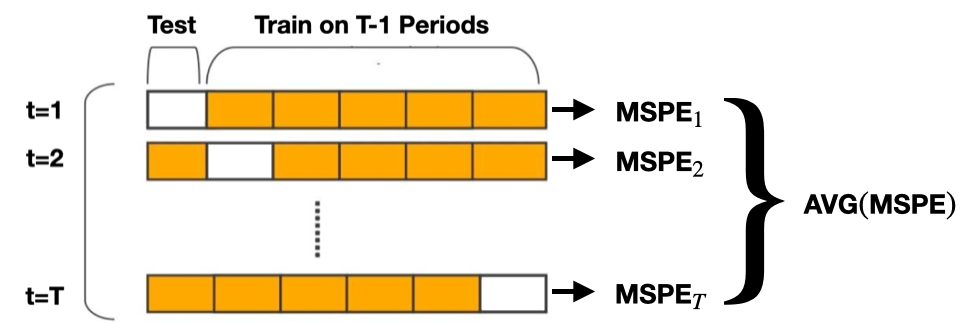}
    \caption{Graphical Representation of Leave-One-Period-Out}
    \label{fig:LOPO}
\end{figure}

\begin{algorithm}[ht]
  \begin{algorithmic}[1]
    \State \textsc{Input:} Data  $D=\{\{Y_{i}^t, X_i^{t^{(m)}} = [\mathcal{W}_{i,t}^{{(m)}},Z_{i}^t, \textbf{f}_i^{(m)}(Z^{t})] \}_{i\in[N],t\in [T]}\}_{m\in[M]}$ where $\mathcal{W}_{i,t}^{{(m)}}$ are other model $m$ specific features
    \For{$m \in [M]$}
    \For{$t \in [T]$}
    \State Estimate $\what{\theta}^{(m)}$ using data $D_{-t}^{(m)}$ (excluding period $t$)
    \State Predict $\hat Y^{t}(\what{\theta}^{(m)}) \gets X^{t^{(m)}}\what{\theta}^{(m)}$ using data $D_{t}^{(m)}$
    \State Store mean squared prediction error: $\MSPE_t(\what{\theta}^{(m)}) \gets \frac{1}{N}\sum_{i=1}^N(\hat{Y}_{i}^t(\what{\theta}^{(m)})-Y_i^t)^2$
    \EndFor
    \State $\text{MODEL\_MSPES}_m\gets \textsc{Average}(\MSPE)$
    \EndFor
    \State Return $\arg\min_{m\in[M]}\text{MODEL\_MSPES}_m$
      \caption{Leave-One-Period-Out Model Selection}
    \label{alg:LOPO_modelselect}
  \end{algorithmic}
\end{algorithm}

When interference is present, outcomes are non-stationary due to the increasing treatment allocation of a roll-out.  For example, an estimator that is fitted on the first period of a roll-out may not extrapolate to the last period. 
Our proposed procedure,~\textsc{Leave-One-Period-Out} (LOPO), leverages the fact that every period offers an opportunity to test an interference model's ability to extrapolate outcomes from different levels of treatment exposure.
In parallel, we leave out each period for testing and estimate parameters $\what{\theta}$ on all other periods. After we have predicted outcomes $\hat Y$ for every period, we compute the MSPE over each prediction task and output the model that minimized the average of these MSPEs. This procedure is visualized in Figure \ref{fig:LOPO} and formalized in Algorithm \ref{alg:LOPO_modelselect}.  While this is only a heuristic, our empirical results in Section \ref{subsec:sim_results} demonstrate its effectiveness.

% Non-EC alg

% EC Version
% \begin{algorithm}
% 	\SetAlgoNoLine\DontPrintSemicolon
% 	\KwIn{Data  $D=\{\{Y_{i,t}, Z_{i,t}, \textbf{f}_i^{(m)}(Z_{t}), W^{(m)}_{i,t} \}_{i\in[N],t\in [T]}\}_{m\in[M]}$}

%         \For{$m \in [M]$}
%         {
%                 \For{$t \in [T]$}{
%                     Estimate $\what{\theta}^{(m)}$ using data $D_{-t}^{(m)}$ (excluding period $t$)\;
%                     Predict $\hat Y_{t}$ using $\what{\theta}$ and $D_{t}^{(m)}$\;
%                     Store mean squared prediction error: $\MSPE_t \gets \frac{1}{N}\sum_{i=1}^N(Y_{i,t}-\hat{Y}_{i,t})^2$\;
%                 }
%             $\text{MODEL\_MSPES}_m\gets \textsc{Average}(\MSPE_t)$\;
%         }
% 	\KwOut{The model, $m$, that solves $\arg\min_{m\in[M]}\text{MODEL\_MSPES}_m$}
%         \caption{Leave-One-Period-Out Model Selection}
%         \label{alg:LOPO_modelselect}
% \end{algorithm}

\begin{table}
\centering
\def\arraystretch{1.5}%
\resizebox{\textwidth}{!}{\begin{tabular}{|p{3.2cm}|p{7.3cm}|p{2.3cm}|p{2.3cm}|}
\hline
\bf Model Selection Method  & \bf Overview                                                                               & \bf Considers Temporal Variation? & \bf Considers Network Structure? \\
\hline
\textsc{No Roll-out}              & Simulates additional periods all with $50\%$ of units treated. Then applies Algorithm \ref{alg:LOPO_modelselect}. & No                            & Yes                          \\
\textsc{Pooled $\mathcal{K}$-Fold} &  $\mathcal{K}$-Fold cross validation over units after pooling all periods together with $\mathcal{K}=10$.                                              & No                            & No                           \\
\textsc{Train First}             & Estimates model on first $T-1$ periods and evaluates model on period $T$.                   & Yes                           & Yes                          \\
\textsc{Train Last}              & Estimates model on last $T-1$ periods and evaluates model on first period.                & Yes                           & Yes                          \\
\textsc{LOPO} (Proposed) & Applies the procedure in Algorithm \ref{alg:LOPO_modelselect}.                                                   & Yes                           & Yes                         \\ \hline
\end{tabular}}
\caption{Model Selection Mechanisms}
\label{tbl:model_select_methods}
\end{table}

\paragraph{Baseline procedures} To evaluate our proposed \textsc{Leave-One-Period-Out} (LOPO) procedure, we compare its performance against reasonable baselines outlined in Table~\ref{tbl:model_select_methods}, including methods that incorporate both temporal and network structure and those that do not.
The \textsc{No-Roll-out} procedure considers what happens when the sample size is increased to match that from a roll-out design, but no temporal variation is generated in the interference structure because treatment status remains constant. The procedure provides a fair comparison against our proposed \textsc{LOPO} method in case the gains our method achieves are simply due to the increased effective sample size from a roll-out. We also include a \textsc{Pooled $\mathcal{K}$-Fold} procedure which pools the data across all periods and conducts standard $\mathcal{K}$-fold cross-validation. \textsc{Pooled $\mathcal{K}$-Fold} evaluates how well our proposed methodology works relative to standard cross-validation tools that implicitly assume all data points are exchangeable. 
Since the exchangeability assumption is violated due to interference,  the \textsc{LOPO} procedure circumvents this problem by exchanging whole periods. Finally, \textsc{Train First} and \textsc{Train Last} are in the same spirit as \textsc{LOPO} method, except they only consider steps $t=T$ and $t=0$ respectively of Algorithm~\ref{alg:LOPO_modelselect}. They preserve both network and temporal structures of the experiment and are less computationally intensive. There are other model selection methods not mentioned here that are similar to \textsc{LOPO}, e.g., train on the first and last periods and evaluate on all other periods. We find that in most cases \textsc{LOPO} achieves the best performance, and omit them from our comparisons.

 \subsection{Simulation Setup}\label{subsec:sim_results}

Continuing from Example~\ref{ex:linearmodels}, we consider a two-sided marketplace between advertisers and auctions. 
We anchor our experiments in a previously motivated setting where there are two competing interference graphs that might describe the observed interactions across advertisers~\citep{brennan2022cluster}. For example, one graph might consider an advertiser's historical spending for a certain window of time, whereas another might consider only certain types of spend or a different window of time. Even if both graphs consider the same bipartite graph to represent interactions between advertisers and auctions, there are different ways to ``fold'' this graph into a one-sided interference network of advertisers with other advertisers, as suggested by \citet{brennan2022cluster}.
Each folded graph leads to different interference neighborhoods, which we can define as $\mathcal{G}_1(\cdot)$ and define $\mathcal{G}_2(\cdot)$. We now generate data according to the \emph{true} model of interference described by $\mathcal{G}_1(i)$ 
\begin{align}
 Y_{i}^t(Z^t) &= \alpha\opt    + \tau\opt   \cdot Z_{i}^t
                    + \eta_1\opt   \cdot \sum_{j \in \mathcal{G}_1(i)} Z_{j}^t  + \epsilon_{i,t}.
\label{eq:true_model}
\end{align}
Let us define several competing models of interference to predict advertisers' spend. We want to test how reliably each model selection method can select the true model. We consider the following competing models,
\begin{align}
    Y_{i}^t(Z^t) &= \alpha\opt    + \tau\opt   \cdot Z_{i}^t
                        + \epsilon_{i,t}\label{eq:baselines_models_1} \\
    Y_{i}^t(Z^t) &= \alpha\opt    + \tau\opt   \cdot Z_{i}^t
                        + \eta_1\opt   \cdot \sum_{j \in \mathcal{G}_2(i)} Z_{j}^t + \epsilon_{i,t} \label{eq:baselines_models_2}
    % Y_{i}(Z^t) &= \alpha + \tau\cdot Z_{i}^t
    %                     + \eta_1\cdot \left(\frac{1}{N_1(i)}\sum_{j \in N_1(i)} Z_{j}^t \right)^2 + \epsilon_{i,t} \\
    % Y_{i}(Z^t) &= \alpha + \tau\cdot Z_{i}^t
    %                     + \eta_1\cdot \frac{1}{N_1(i)}\sum_{j \in N_1(i)} Z_{j}^t
    %                     + \eta_2\cdot \left(\frac{1}{N_1(i)}\sum_{j \in N_1(i)} Z_{j}^t \right)^2 + \epsilon_{i,t}  
  \end{align}
The  model~\eqref{eq:baselines_models_1} assumes no interference and the model~\eqref{eq:baselines_models_2} considers an incorrect interference network defined by $\mathcal{G}_2(\cdot)$ as in Example \ref{ex:linearmodels}. 

In real-world applications, the effect of interference is typically thought to be smaller than the direct effect of treatment \citep{BlakeCoey2014Experimentation}. To capture this, we set $\tau\opt    = 5$, $\eta_1\opt    = 2$, and simulate data using $\epsilon_{i, t} \simiid N(0, 1)$. In each experiment, we observe outcome and treatment data from a $50\%$ completely randomized roll-out (cf. Def.~\ref{def:CRD}), such that $50\%$ of units are treated after the last period. We fit a linear regression model associated with the selected model to estimate the total treatment effect. Throughout, we assume $T=5$ in each experiment which coincides with many practical applications with few roll-out periods.

\begin{table}[t] %[!hbtp] \centering 
\resizebox{\textwidth}{!}{
\begin{tabular}{@{\extracolsep{5pt}} P{20mm}P{22mm}P{22mm}P{22mm}P{22mm}P{22mm}} 
\\[-1.8ex]\hline 
\hline \\[-1.8ex] 
Figures & \textsc{No Roll-out} & \textsc{Pooled $\mathcal{K}$-Fold} & \textsc{Train First} & \textsc{Train Last} & \textsc{LOPO} (Proposed) \\ 
\hline \\[-1.8ex] 
Figure 5a & 81.2\newline
(80.4, 82.0) & 15.2\newline
(14.5, 15.9) & 51.6\newline
(50.6, 52.6) & 47.8\newline
(46.8, 48.8) & \textbf{12.6}\newline
\textbf{(12.0, 13.3)} \\ 
Figure 5b  & 81.4\newline
(80.6, 82.2) & 19.6\newline
(18.8, 20.4) & 45.8\newline
(44.8, 46.8) & 46.2\newline
(45.2, 47.2) & \textbf{15.0\newline
(14.3,15.7)} \\ 
Figure 6a  & 21.8\newline
(21.0, 22.6) & \textbf{19.0\newline
(18.2, 19.8)} & 39.6\newline
(38.6, 40.6) & 44.8\newline
(43.8, 45.8) & \textbf{19.0\newline
(18.2, 19.8)} \\ 
Figure 6b  & 45.8\newline
(44.8, 46.8) & \textbf{11.0\newline
(10.4, 11.6)} & 19.8\newline
(19.0, 20.6) & 18.6\newline
(17.8, 19.4) & 19.6\newline
(18.8, 20.4) \\ 
\hline \\[-1.8ex] 
\end{tabular}}
\caption{Percentage of Models Incorrectly Selected: Each row displays the percentage of incorrect model selections by each procedure for the simulation in the corresponding figure. 95\% bootstrapped confidence intervals are displayed in parentheses.} 
  \label{tbl:model_selection_results} 
\end{table} 

% ------------------------
\subsection{Simulation Results}

In the experiments below, we consider several variations for establishing the effectiveness of the \textsc{LOPO} methodology. We first consider both even and uneven $50\%$ roll-outs. We then consider introducing additional interference terms to the true model and allowing for individual heterogeneity and time-varying effects. Lastly, we consider the effect of network sparsity by evaluating the performance of our model selection mechanisms as the underlying network density increases.

 \begin{figure}[t]
    \centering
    \subfigure[Even Roll-out]{\includegraphics[width=0.49\textwidth]{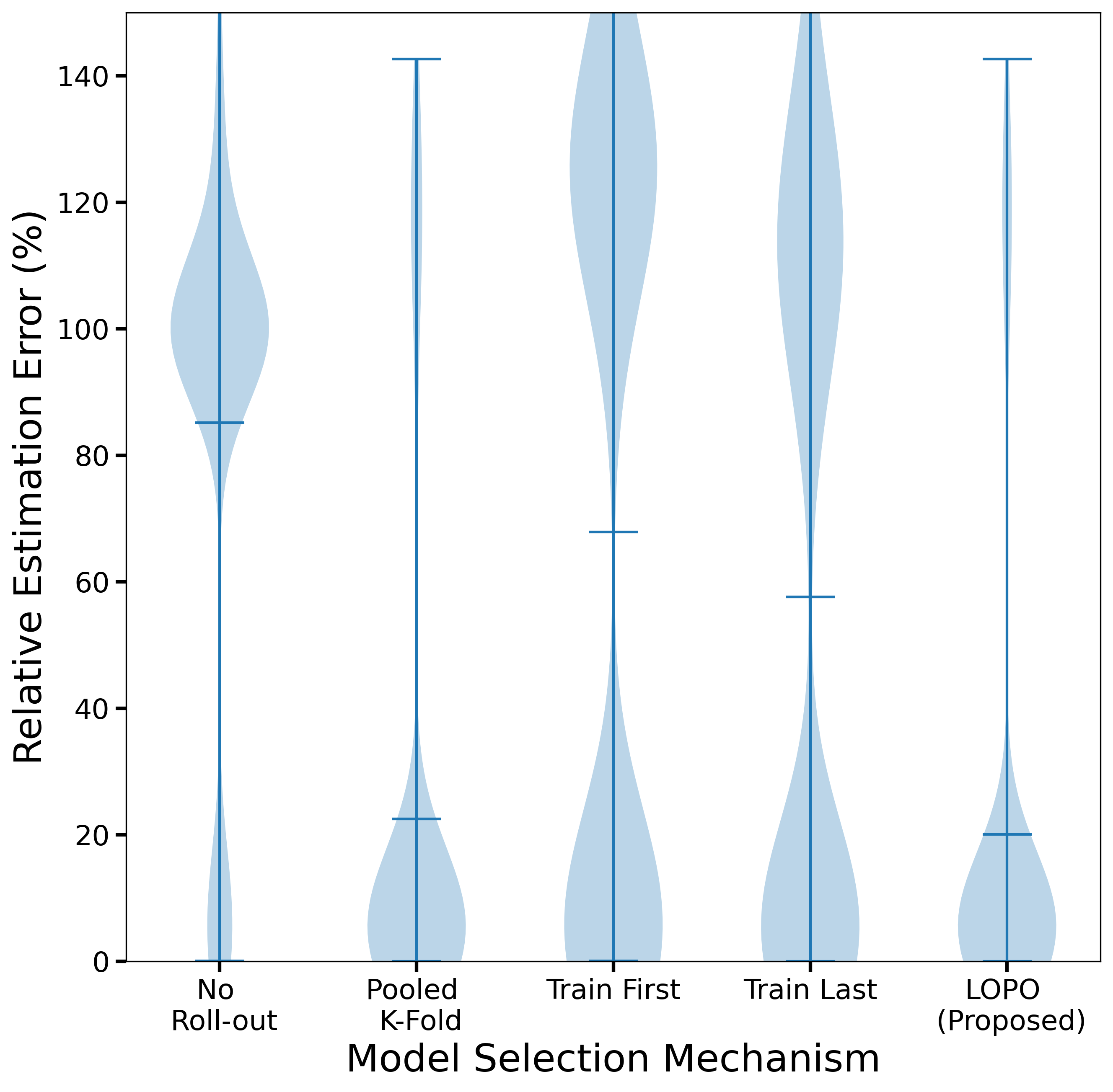}}
    \subfigure[Uneven Roll-out]{\includegraphics[width=0.49\textwidth]{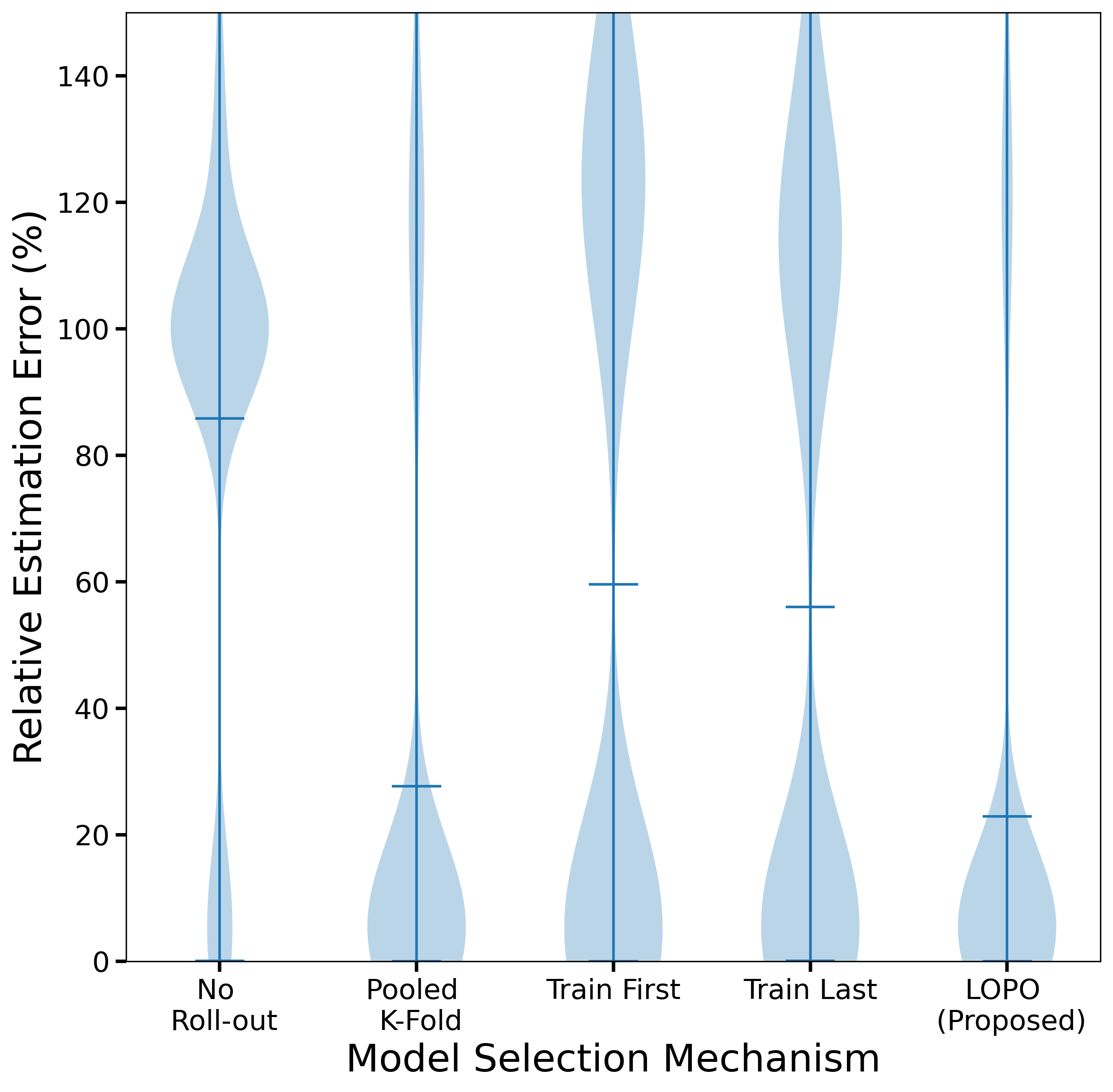}}
    \caption{Even vs. Uneven Roll-outs: figures are generated by averaging the results of 500 experiments. Each model is estimated with a sample size of $N=1000$ and $T=5$ periods. The plots show the distribution of relative estimation error (\%) for the model selection mechanism. The central tick marks represent the median.}
    \label{fig:sim_even_uneven}
\end{figure}

We evaluate each method on two metrics: how often it selects the correct model of interference and how well it minimizes the estimation error of the TTE regardless of whether it has chosen the correct model. Table \ref{tbl:model_selection_results} summarizes the percentage of times each model selection procedure selects the \emph{incorrect} model. However, since an incorrectly selected model could still yield similar estimates of the TTE, Figures \ref{fig:sim_even_uneven} and \ref{fig:sim_DGP} present the distribution of the relative absolute percent estimation error of the $\TTE$ over $500$ runs. 
%Choosing the incorrect model also impacts the definition of $c$ and so this difference is also accounted for when estimating the $\TTE$ under the selected model.

\paragraph{Even vs. uneven roll-outs.} In this first experiment, we consider model selection in the even vs. uneven roll-out setting. In the even setting, we consider treatment increments of $10\%$. The uneven setting considers five treatment periods with proportions, $\vec{p}=[0.01, 0.09, 0.10, 0.15, 0.15]$, that sum to $0.50$. We observe the proposed \textsc{LOPO} procedure performs the best in both the even and uneven settings, verifying  our intuition outlined above. The \textsc{Pooled $\mathcal{K}$-Fold} procedure also does a good job at minimizing estimation error but performs worse than \textsc{LOPO} on selecting the correct model, as we observe in rows one and two of Table \ref{tbl:model_selection_results}. This may emphasize the importance of considering the underlying network structure in our selection procedure. In contrast, \textsc{Train First}, \textsc{Train Last}, and \textsc{No Roll-out} have very different performances across simulation settings and usually underperform relative to \textsc{LOPO}. \textsc{LOPO} tends to be more robust to changes in the roll-out schedule, which is useful when the practitioner cannot control the roll-out schedule. %Intuitively, \textsc{LOPO} minimizes MSPE using all periods as tests and thus is able to account for the changes in added variation between periods which occurs in the uneven roll-out. This is .

\begin{figure}[t]%[hbtp][t]
    \centering
    \subfigure[Neighbors of Neighbors Interference]{\includegraphics[width=0.49\textwidth]{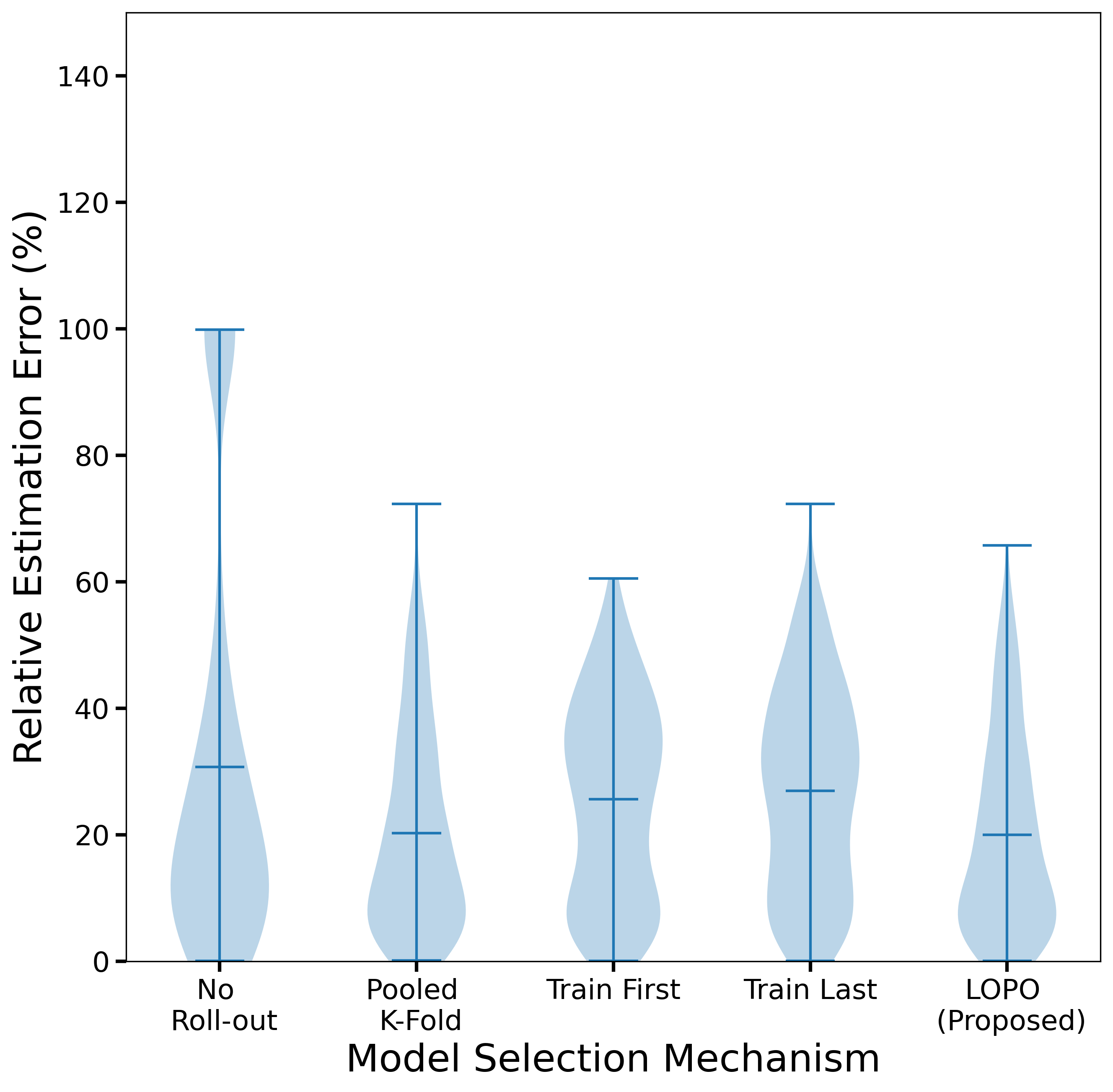}\label{fig:sim_2orderneighbors}}
    \subfigure[Individual Heterogeneity and Time Effects]{\includegraphics[width=0.49\textwidth]{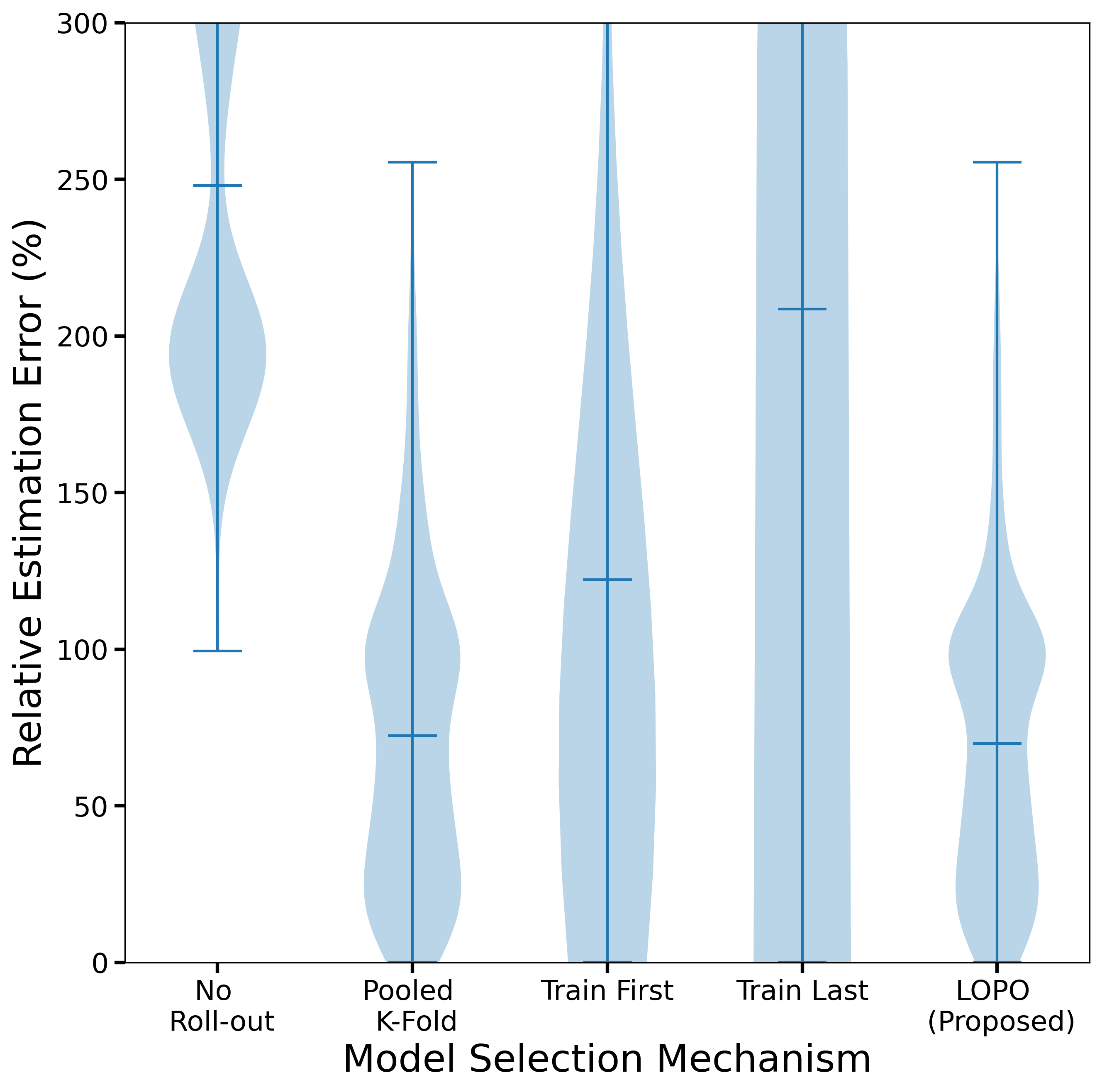}\label{fig:sim_FETE}}
    \caption{Varying the True Interference Model: Figures are generated by averaging the results of 500 experiments. Each model is estimated with a sample size of $N=1000$ and $T=5$ periods. The coefficient on 2nd-order neighbor interference is $\eta_2\opt    = 2$. The plots show the distribution of relative estimation error (\%) for the model selection mechanism. The central tick marks represent the median.}
    \label{fig:sim_DGP}
\end{figure}

\paragraph{Varying the true interference model} We now consider  variations to the data-generating model. Figure \ref{fig:sim_2orderneighbors} considers a true potential outcome model where interference comes from both first-order and second-order neighbors
\begin{equation}
 Y_{i}^t(Z^t) = \alpha\opt    + \tau\opt   \cdot Z_{i}^t
                        + \eta_1\opt   \cdot \sum_{j \in \mathcal{G}_1(i)} Z_{j}^t +
                        \eta_2\opt   \cdot \sum_{k \in \mathcal{G}_1^{(2)}(i)} Z_{k}^t
                        + \epsilon_{i,t}\,,
\end{equation}
where $\mathcal{G}_1^{(2)}(i)$ defines the set of neighbors-of-neighbors of unit $i$ under $\mathcal{G}_1(\cdot)$. Figure \ref{fig:sim_FETE} considers the performance of our model selection methods when adding individual heterogeneity and time-varying effects to the true potential outcomes model:
\begin{equation}\label{eq:FE_model}
 Y_{i}^t(Z^t) = \alpha_i\opt    + \gamma_t\opt    + \tau\opt   \cdot Z_{i}^t
                        + \eta_1\opt   \cdot \sum_{j \in \mathcal{G}_1(i)} Z_{j}^t
                        + \epsilon_{i,t}
\end{equation}
To make the comparison fair, we add these additional terms of the true models in Figure \ref{fig:sim_DGP} to all the alternate interference models specified in \eqref{eq:baselines_models_1}-\eqref{eq:baselines_models_2}. In each of these experiments, we use an even roll-out with a $10\%$ per period increment. 

In Figure \ref{fig:sim_2orderneighbors}, all procedures tend to perform better in estimation error when adding the second order neighbors interference term. This is likely because outcomes are now more correlated with the underlying network since interference now has a larger spillover effect. On the other hand, all model selection procedures tend to do worse when considering unit-fixed and time-fixed effects. This is unsurprising since the additional individual and time-varying terms make it difficult to distinguish between spillover effects and individual and temporal heterogeneity. In Figure \ref{fig:sim_FETE}, we observe both \textsc{LOPO} and \textsc{Pooled $\mathcal{K}$-Fold} perform better in terms of estimation error relative to the other baselines.

\begin{figure}[t]
    \centering
    % \subfigure[Misspecified Interference Network]{\includegraphics[width=0.50\textwidth]{figures/overall_success_google_graph_FalseFold.png}\label{subfig:sim_false_fold}}
    %\subfigure[Varying Network Sparsity]{
    \includegraphics[width=0.7\textwidth]{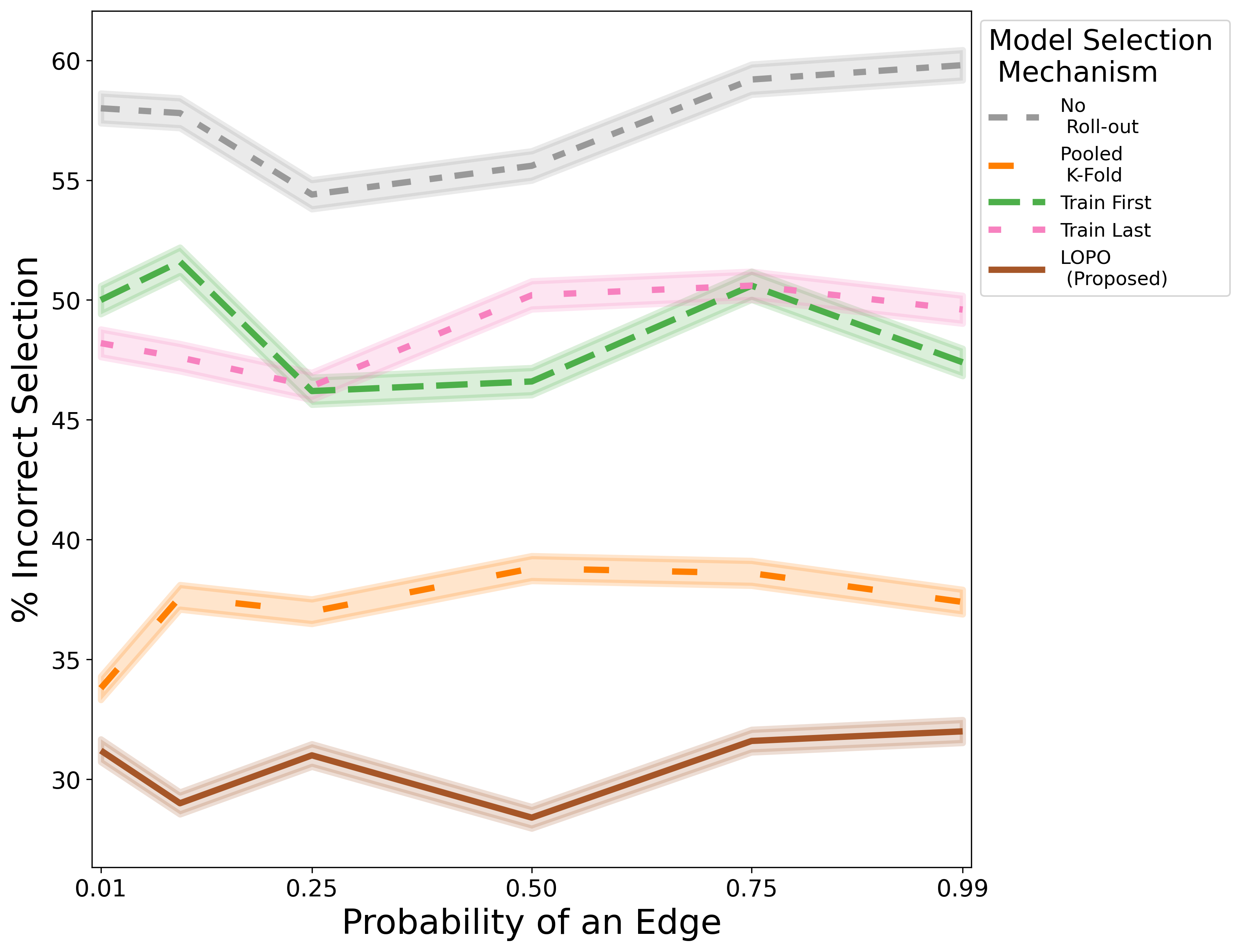}
    %}
    \caption{Varying Network Sparsity: Figures are generated by averaging the results of 500 experiments. Each model is estimated with a sample size of $N=1000$ and $T=5$ periods. 95\% bootstrapped confidence sets are displayed by the shaded region. We exclude probabilities of $0$ and $1$ due to multicollinearity issues.}
    \label{fig:sim_sparse}
\end{figure}

\paragraph{Network sparsity} Finally, we consider the performance of our procedure when we vary the sparsity of the underlying network. As we have seen in Section~\ref{sec:Identification}, sparsity can greatly influence the likelihood that the $\TTE$ is identified. Naturally, we expect this parameter to also influence model selection. For example, we might expect graphs that are very sparse to generate little variation preventing us from learning how to extrapolate treatment exposures to outcomes. On the other hand, graphs that are very dense tend to generate colinear data, which complicates parameter estimation. 

In Figure \ref{fig:sim_sparse}, we generate a series of Erdos-Renyi graphs with an increasing probability of any two units having an edge, using each graph to define $\mathcal{G}_1(\cdot)$. Figure \ref{fig:sim_sparse} displays how often each model selection procedure fails to pick out the true model as we increase this probability. We observe that the procedures are not very sensitive to the underlying sparsity of the graph,  with relatively flat selection rates across each graph. Strikingly, the \textsc{LOPO} procedure outperforms every other model selection procedure in selecting the true model at every point.

\paragraph{Polynomial Models} In the previous experiments we have generated potential outcomes based on linearly additive models. However, the \textsc{LOPO} procedure can be extended to non-linear outcome models as well. We illustrate this using the polynomial outcome model of~\citet{yuetal2022rollout}, which assumes the data-generating process
\begin{equation}\label{eq:PI_simulation}
    Y_i(\vec{z}) = c_{i, \emptyset} + \sum_{j \in \mathcal{N}(i)}\tilde{c}_{i,j} z_{j} + \sum_{l=2}^\beta \left(\frac{\sum_{j \in \mathcal{N}(i)}\tilde{c}_{i,j} z_{j}}{\sum_{j \in \mathcal{N}(i)}\tilde{c}_{i,j}} \right)^l,
\end{equation}
where $c_{i,\emptyset} \in U[0,1]$, $\tilde{c}_{i,i} \sim U[0,1]$, and for $i\neq j$ $\tilde{c}_{i,j} = v_j |\mathcal{N}(i)| / \sum_{k : (k,j)\in E}|\mathcal{N}(k)|$ and $v_j \sim U[0,r]$. Here, $r$ is a parameter that controls the magnitude of indirect effects; in our simulations, we set $r=2$. Like \cite{yuetal2022rollout} we do not include a noise term so that any error is due to the misspecification of the model. We define $\mathcal{N}(i)$ according to a sparse Erdos-Renyi graph where the probability of an edge between any two nodes is $0.1$.

An important consideration in \cite{yuetal2022rollout} is the choice of $\beta$ controlling the number of higher-order polynomial terms. \cite{yuetal2022rollout} take a design based perspective where they choose $T=\beta$ periods of roll-out based on a known $\beta$. Instead, we consider how a researcher might select $\beta$ given a $T$-period roll-out. This complementary perspective is important in many online platforms where roll-outs are frequently implemented independent of the model specification; researchers are often tasked with evaluating the effects of an intervention ex-post. Figure \ref{fig:select_beta} shows the results of an experiment where we select $\beta$ by applying $LOPO$ to four variations of the interference model \eqref{eq:true_model} including a model with no interference terms, a second-order term, and a third-order term. $\beta$ can then be inferred by inspecting how many higher-order terms are in the selected model. 
\begin{figure}[ht]
    \centering
        \subfigure[Selecting the Polynomial Order]{\includegraphics[width=0.49\textwidth]{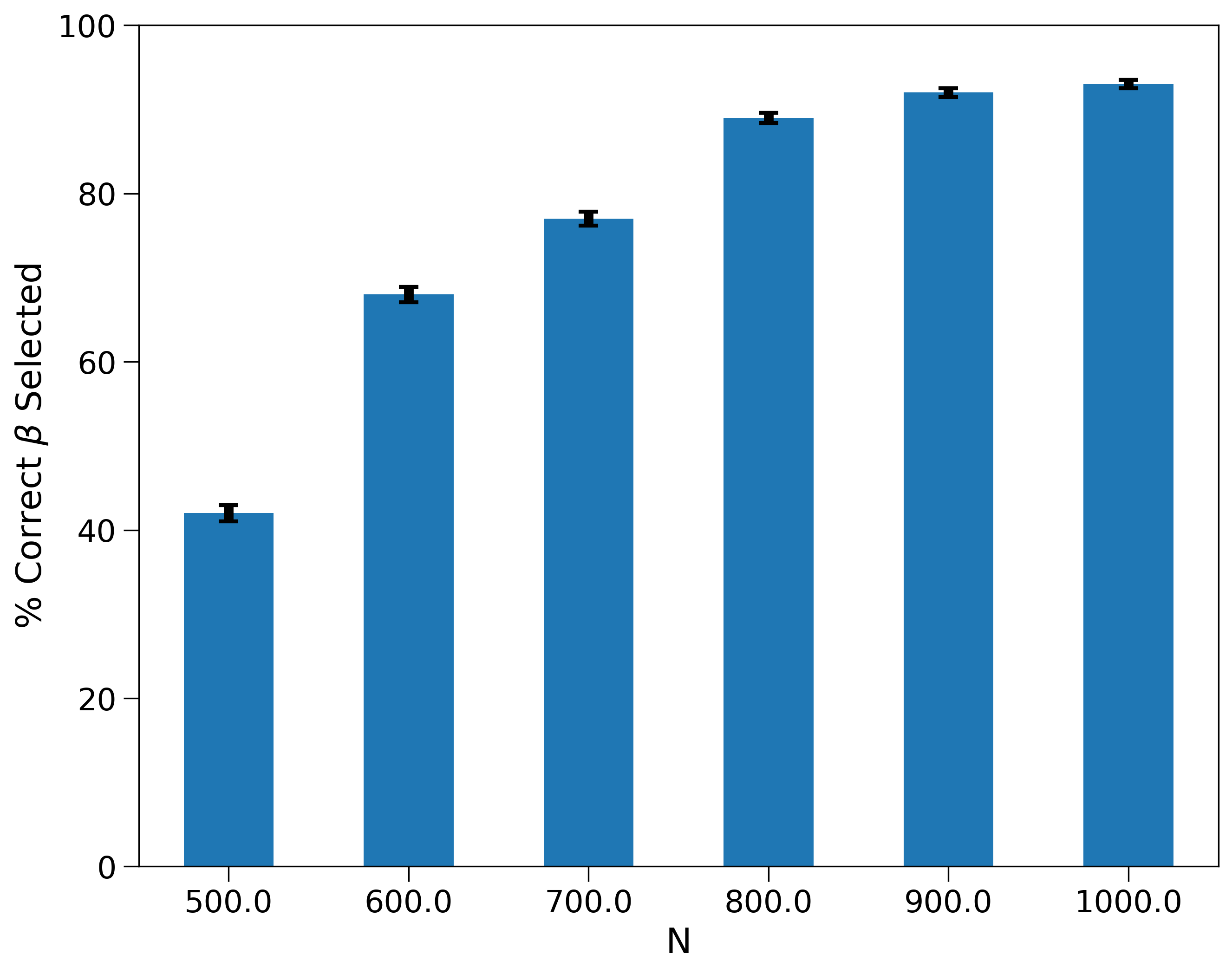}}
    \subfigure[Relative Bias]{\includegraphics[width=0.49\textwidth]{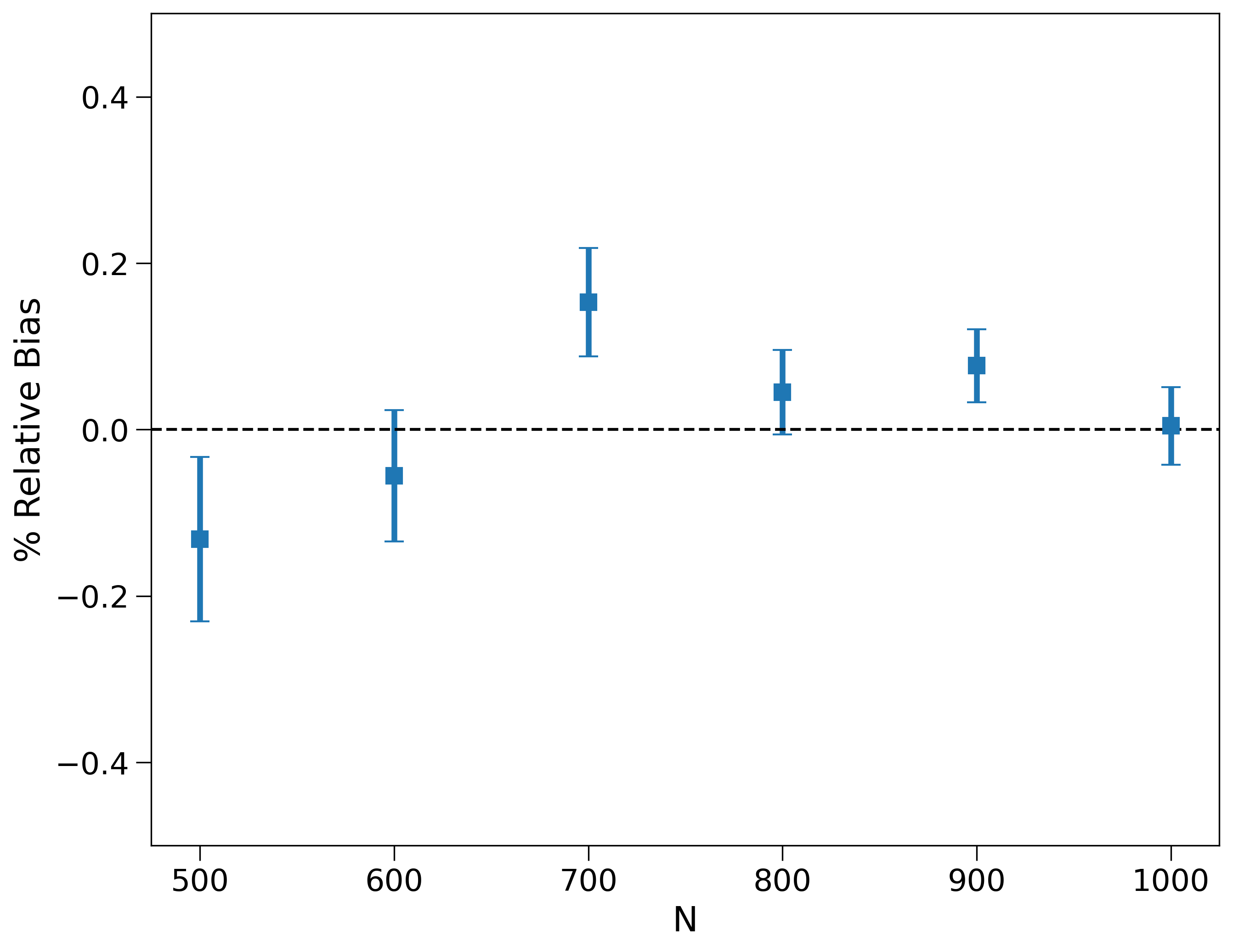}}
    \caption{Selecting the Degree of Polynomial Models: figures are generated by averaging the results of 100 experiments. In each experiment, sample size is given by the x-axis, $N$, with $T=4$ and an even roll-out. We display bootstrapped $95\%$ confidence intervals of the relative bias. We use the DGP given by \cite{yuetal2022rollout} in Eq. \eqref{eq:PI_simulation} with $r=2$ and $\beta=2$.}
\label{fig:select_beta}
\end{figure}

Notably, \textsc{LOPO} predominantly selects the correct model associated with $\beta=2$, and the rate of accurate selections rises with growing values of $N$. After determining the appropriate $\beta$ value, researchers can employ the Lagrange interpolation estimation technique with $T=\beta$, as elaborated in Section 3 of \cite{yuetal2022rollout}. The results of this procedure are displayed in the second panel of Figure \ref{fig:select_beta} where we find that applying this procedure yields minimal relative bias which is vanishing with the sample size $N$. Additionally, in Section \ref{sec:misspec}, we delve into the impact of model misspecification on the bias of the TTE in the polynomial model.

\subsection{Limitations and Extensions}
\label{subsec:limitations}
\textsc{LOPO} tends to lead to better model selections and lower absolute error than other methods across each experimental setting we have considered, especially relative to the \textsc{No Roll-out} procedure. Our empirical findings align with our theoretical study to come, which quantifies the statistical efficiency gains due to roll-outs and confirms our conjecture that temporal variations can be used to better model interference effects. Interestingly, the \textsc{Pooled $\mathcal{K}$-Fold} procedure performs very well, even though it does not consider the underlying network structure. One explanation may be that our interference networks induce neighborhoods that are relatively uniform across units. In this  case, $\mathcal{K}$-Fold cross-validation does not risk leaving out highly central units; uniformity effectively implies  the particular network we used to generate outcomes satisfies exchangeability. %\jpa{not sure I totally get this sentence, could you try to rewrite it}. % Furthermore, the \textsc{LOPO} procedure can be thought of as a specific instance of $\mathcal{K}$-Fold cross-validation and so we should expect some similarity in the performance of the two methods. Still, our \textsc{LOPO} procedure tends to outperform this method in each setting, especially with regards to selecting the correct model. 

While the LOPO procedure is reliable overall, there are instances in which it is outperformed by other procedures. As with any model selection procedure, we can always construct adversarial examples in which \textsc{LOPO} will fail. For example, if all the variation in treatment exposure occurs in the first and last periods, we would expect a procedure that only considers estimation and prediction on these periods to do very well while LOPO---which equally weights the middle period----could perform worse. A more challenging example might be the case of threshold interference. In all of our models, we have considered roll-outs only to $50\%$. If there is a thresholding effect whereby interference only appears at the $51\%$ treated level, then all of our procedures would fail since we would have no variation in the interference effect to learn from.

There are possible extensions to our suggested procedure. For instance, we could consider a weighted average of the MSPEs from each fold. Weights could be proportional to the number of treated observations in the test set, for instance, so that variation in each period is considered. Another approach might be to consider testing on multiple periods at once, e.g., leaving out all combinations of two periods, training on the remaining periods, and testing on these two periods. Such a procedure would maximize the variation we exploit and is computationally very costly for potentially marginal benefit. Still, for any of these extensions, an adversarial example is possible, and we believe the LOPO procedure is a reasonable choice, being both intuitive and reliable in a wide range of environments.

\label{sec:Identification}
\section{Model Identification and Estimation}
\label{sec:Identification}

In the previous section, we discussed how a model of interference may be selected using roll-out experiments. We now take this selected model as given and  characterize how roll-outs allow us to identify causal effects. We begin by providing conditions under which the total treatment effect~\eqref{eq:TTE} can be identified under the presence of interference. We demonstrate theoretically and empirically that roll-outs help satisfy these identification conditions. Furthermore, we prove that when these conditions are met, roll-outs provide gains in statistical efficiency for estimating the TTE. 

\subsection{Potential Outcomes Model and Estimation Framework}
Causal estimation under interference requires structural assumptions. To ground our study, we consider the following model class and associated estimator. 
\paragraph{Linear additive models} We assume potential outcomes are linearly additive in $z_i$ and a known $p$-dimensional feature
vector ${\bf f}_{i}:\left\{ 0,1\right\} ^{N}\to\mathbb{R}^{p}$
\begin{equation}\label{eq:additive-model}
Y_{i}^t(\vec{z}_t)= \alpha_i\opt    + \gamma_t\opt    + \psi_{g(i,t)}\opt    + \tau\opt    \cdot z_{i,t} + \eta^{\star\top} \textbf{f}_i(\vec{z}_t) + \epsilon_{i,t}.
\end{equation}
The class~\eqref{eq:additive-model} allows flexible modeling of interference effects. It subsumes commonly studied models such as exposure mapping and two-way fixed effect models~\citep{Harshaw2023Exposure}, as well as models from Example \ref{ex:linearmodels}.
\begin{itemize}
\item $\alpha\opt    \in \mathbb{R}^N$ is a unit-fixed effects allowing for individual heterogeneity. 
\item   $\gamma\opt    \in \mathbb{R}^T$ models time-varying trends through period-fixed effects. 
\item  $\psi\opt    \in \mathbb{R}^G$ models further unit-period heterogeneity where $g:[N]\times[T]\to [G]$. 
To make estimation tractable, we assume $G<NT$ is small enough such that we never have more parameters than observations. 
\item $\tau\opt    \in \mathbb{R}$ is the direct treatment effect for unit $i$. 
\item $\eta\opt   \in \mathbb{R}^p$ models indirect treatment effects due to interference. 
\item $\{\epsilon_{i,t}\}_{i\in [N],t \in [T]}$ are idiosyncratic noise, for which we will consider different regimes. 
\end{itemize}

Letting $K=N+T+G+p+1$, define the vector of model parameters
\begin{equation}
\label{eq:parameters}
    \theta\opt    \defeq [\alpha\opt   , \gamma\opt   , \psi\opt   , \tau\opt   , \eta\opt   ]\in\R^K,
\end{equation} 
assuming a normalization where $f_i(\mymathbb{0})=0_p$ for all $i$. Under the data generating process~\eqref{eq:additive-model}, the total treatment effect~\eqref{eq:TTE} can be rewritten as the linear combination
 \begin{equation}\label{eq:tte_linear}
     \TTE = c^\top\theta\opt   ~~~\mbox{where}~~~ c = \left[0_N, 0_T, 0_G, 1, \overline{f(\mymathbb{1})}
    \defeq \frac{1}{N}\sum_{i=1}^N f_i(\mymathbb{1}) \right]\in\R^K.
 \end{equation}

\paragraph{Estimation approach} To estimate the TTE under the above linearly additive model~\eqref{eq:additive-model}, we turn to simple linear regression estimators. We define our matrix of covariates for a single roll-out period as
\begin{equation}
\label{eq:features_period}
X^t = [I_N, \ones^\top_t, \ones^\top_{g(i,t)}, Z^t, f(Z^t)]
\end{equation}
where $I_N$ is the $N\times N$ identity matrix representing the individual effects, $\ones^\top_t \in \mathbb{R}^{N\times T}$ is a matrix indicating if an observation belongs to period $t$, $\ones^\top_{g(i,t)}\in \mathbb{R}^{N\times G}$ indicates if observation $i$ is in cluster $g=1,\dots, G$ at period $t$. 
Letting $X=[X^1, \dots, X^T]^\top$ and recalling the definition of coefficients $c$ and parameters $\theta$ defined in Eq.~\eqref{eq:parameters},  we study the linear regression estimator 
\begin{equation}
\label{eq:tte_estimator}
\hat\TTE = c^\top \what{\theta} 
~~~\mbox{where}~~~\what{\theta} = (X^\top X)^{-1}X^\top Y
\end{equation}
whenever $X^\top X$ is non-singular, in which case $\hat\TTE$ is a consistent and unbiased estimator. A major benefit of roll-outs is that they increase the likelihood that $X^\top X$ will be non-singular, which is necessary for the total treatment effect to be identifiable and for $\hat\TTE$ to have the aforementioned statistical guarantees.

\subsection{Model Based Identification}

We now consider conditions that allow us to identify the total treatment effect in the finite population setting and show how they are tied to $X^\top X$ being invertible. We first consider the case of a single interference term to build intuition, and then provide a more general condition for the identification of the TTE. We conclude by showing empirically how roll-outs increase the likelihood that these identification conditions hold in-sample.

In what follows we consider parametric identification of the TTE in the setting of finite populations. We say that a parameter is identified in this sense if there exists a consistent estimator for that parameter where consistency is considered with respect to increasing population sizes. Results of this nature are derived in Section \ref{subsec:est}. In the case of linear models with exogenous covariates, as in our setting, a sufficient condition for identification of the parameter vector is non-singularity of the design matrix, $X^\top X$ (see Section 4.2.1 of \cite{Wooldridge2023CrossPanel}). Section \ref{subsec:param_id} provides sufficient conditions for the non-singularity condition to hold.

\subsubsection{Identification with a Single Interference Term}
\label{subsec:param_id}

Proposition \ref{prop:neighbor_spillovers} below introduces a sufficient condition that ensure $X^\top X$ is invertible when $f_i(\vec{z})\in\mathbb{R}$. We state the condition in the language of interference networks to illustrate how they relate to roll-outs and interference. The key idea here is that if we can find some unit in the control group connected to a treated unit and observe the spillover effect on this individual, then we have enough information about the interference mechanism to extrapolate to the case of total treatment. Roll-outs increase the probability that this sufficient condition holds by increasing the proportion of treated individuals in each period. The proof of this result is given in Section \ref{subsec:pf_prop_spill}.
\begin{proposition}\label{prop:neighbor_spillovers}  Consider a $T>1$ period roll-out under the following linearly additive model
$$Y_i^t(\vec{z})= \alpha\opt    + \tau\opt    \cdot z_i + \eta\opt    \cdot f_i(\vec{z}) + \epsilon_{i,t}.$$
Assume there are $t,t' \in [T]$, $i,j\in[N]$,  $(i, t)\neq (j, t')$, such that $Z^t_i = Z^{t'}_j=0$, ${f}_i({Z}^t)\neq 0$, and ${f}_i(Z^t)\neq {f}_j(Z^{t'})$.
Then $X^\top X$, as defined in Eq.~\eqref{eq:features_period}, is non-singular.
\end{proposition}
\noindent The  condition ${f}_i(Z^t)\neq {f}_j(Z^{t'})$ ensures we observe variation in the interference term so that, as the roll-out progresses, the interference effects vary. The condition ${f}_i({Z}^t)\neq 0$ further ensures that we observe a spillover effect on an untreated unit. Together, these conditions are sufficient (but not necessary) for identifying the TTE by ensuring the invertibility of the Gram matrix, $X^\top X$. 

Next, we apply Proposition~\ref{prop:neighbor_spillovers} in the context of the interference graph from Example~\ref{ex:linearmodels}.
\begin{corollary}
\label{cor:spillover_identification} 
Consider the model from Proposition~\ref{prop:neighbor_spillovers} where the interference term is given by the model \eqref{eqn:N1-model} in Example~\ref{ex:linearmodels}, i.e., $f(\vec{z}) = \sum_{j \in \mathcal{G}_1(i)} \vec{z}_{j}$. Assume that neighbors are commutative so that $i \in \mathcal{G}_1(j)$ implies $j \in \mathcal{G}_1(i)$.
If at time $t>1$, there is a treated unit $j$ with an untreated neighbor ($i \in \mathcal{G}_1(j)$ with $Z_i^t=0$),  $X^\top X$ is non-singular.
\end{corollary}

\subsubsection{Identification with General Interference}

In practice, we are only interested in estimating the TTE, not the entire parameter vector $\theta$ in the model~\eqref{eq:additive-model}. Recalling the linear representation~\eqref{eq:tte_linear} for the TTE,
we now show we can identify the TTE even when the individual components of $\theta$ are not identifiable.
Our result shows that the TTE is identified so long as  the linear transformation $c$ that maps $\theta$ to the TTE~\eqref{eq:tte_linear} lies in the space spanned by the covariates~\eqref{eq:features_period}.
Intuitively, this shows that we can identify the TTE under general interference patterns whenever
 the linear transformation~\eqref{eq:tte_linear} can be represented by the observed data. 
This is particularly useful in the small $N$ and $T$ regime where there may not be enough variation to compute a typical least squares estimate. 
\begin{theorem}\label{thm:uniqueness} Under the data-generating model~\eqref{eq:additive-model}, recall the linear transformation  $c$ that maps the vector of parameters to the TTE~\eqref{eq:parameters}. If $c\in\text{span}(X^{\top})$, then $\{c^\top\theta:X^{\top}X\theta=X^{\top}Y\}$ is a singleton.
\end{theorem}

To clarify the importance of Theorem \ref{thm:uniqueness}, which we prove in Section \ref{subsec:pf_thm_id}, we consider the following example where $X^\top X$ is singular but the TTE is still well-defined and identifiable via Theorem \ref{thm:uniqueness}.
\vspace{-12pt}
\begin{figure}[H]
\centering
% \vcenter{\hbox{
% \begin{tikzpicture}
% \begin{scope}[every node/.style={circle, thick, draw}]
%     \node (A) at (0,1) {A};
%     \node (B) at (2,1) {B};
%     \node (C) at (1,0) {C};
% \end{scope}

% \begin{scope}[>={[black]},
%               every edge/.style={draw=red,very thick}]
%     \path [->] (A) edge node {} (B);
%     \end{scope}
% \end{tikzpicture}}}
% \qquad\qquad
% $$X = \begin{bmatrix} 1 & 0 & 0 \\ 1 & 0 & 0 \\ 1 & 0 & 0 \\ 1 & 1 & 1 \\ 1 & 1 & 1 \\ 1 & 0 & 0 \end{bmatrix}$$
\includegraphics[width=0.5\textwidth]{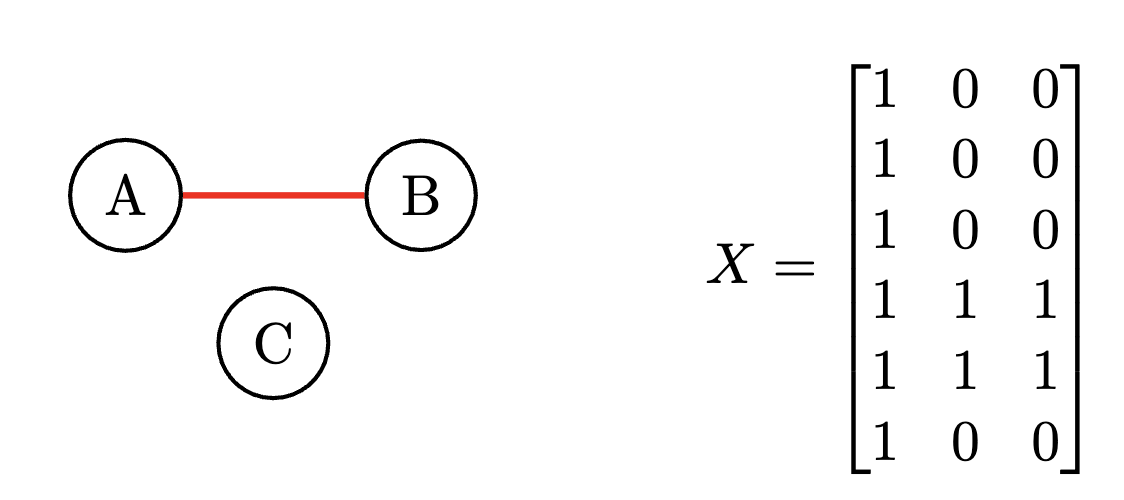}
\caption{A disconnected network defined by $\mathcal{G}(\cdot)$ and corresponding feature matrix $X$ defined by \eqref{eq:ex2}.}\label{fig:ex2_network}
\end{figure}

\begin{example}[Identifying the TTE when $\theta$ is \emph{not} identified]\label{ex:tte_identified_theta_not} Consider the linear interference model of Example \ref{ex:linearmodels} with no individual heterogeneity, so that $\forall~ i \in [N],~\alpha_i = \alpha$,
\begin{equation}\label{eq:ex2}
 Y_i^t(Z^t) = \alpha\opt    + \tau\opt   \cdot Z_{i}^t
 + \eta\opt   _1\cdot \sum_{j \in \mathcal{G}(i)} Z_{j}^t + \epsilon_{i,t}
\end{equation}
Let $N=3$ and define the interference network $\mathcal{G}(\cdot)$ by the graph in Figure~\ref{fig:ex2_network}. Suppose that $T=2$, such that in the first period $t=1$ no individuals are treated, and that in period $t=2$ we treat observations $A$ and $B$, generating the feature matrix $X$ in Figure \ref{fig:ex2_network}. 
We wish to estimate the TTE using the correctly specified model~\eqref{eq:ex2}. Notice that $X^\top X$ is singular since because columns 2 and 3 of $X$ are linearly dependent. 
The TTE in this model is given by $\tau + \eta$ implying $c = [0,1,1]$, and $c \in \text{span}(X^\top)$ because $X^\top v  = c$ for $v=[0,0,0,0,1,-1]^\top \in \R^{6}$.  Theorem~\ref{thm:uniqueness}  shows that the TTE is identifiable; in particular, we can estimate TTE by looking at the difference in outcomes for unit $A$ or $B$ at periods $t=1$ and $t=2$.
\end{example}

Proposition \ref{prop:neighbor_spillovers} and Theorem \ref{thm:uniqueness} consider when it is possible to identify a linear combination of parameters from a linear regression. While it is possible to satisfy the conditions of Proposition \ref{prop:neighbor_spillovers} and Theorem \ref{thm:uniqueness} with variation in interference effects over a single period---sometimes called spatial variation~\citep{aronow2017estimating}---in many cases, we also need temporal variation to achieve identification, which roll-outs provide.  For example, when we have individual heterogeneity parameters, $\{\alpha_i\}_{i\in[N]}$, temporal variation in individual responses is required for identification. 
\begin{figure}
\centering
    \includegraphics[width=0.5\textwidth]{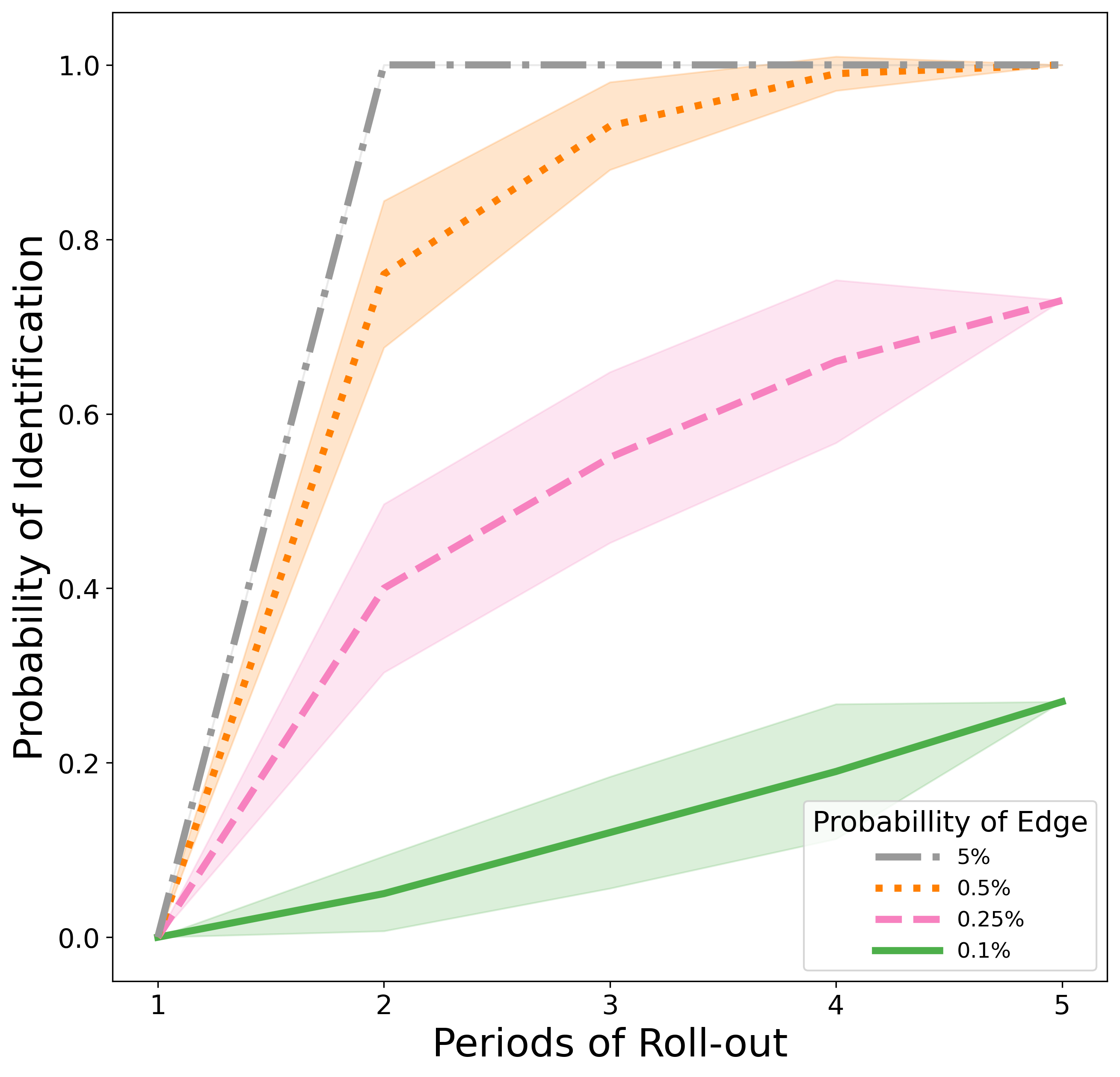}
    \caption{Probability that $c\in \text{span}(X^\top)$ for an underlying Erdos-Renyi graph with varying edge-formation probabilities of an edge. Probabilities are computed over 100 experiments. We display 95\% CIs. Over all realizations with a 5\% probability of an edge (shown in grey) results are constant, hence there is no visible CI.}
\label{fig:prob_id}
\end{figure}
Roll-outs should not only help us in identifying individual effects but also enable us to identify interference effects as well as the total treatment effect. Specifically, roll-outs provide added variation that increases the probability that the conditions such as the ones outlined in Proposition \ref{prop:neighbor_spillovers} and Theorem \ref{thm:uniqueness} hold. Figure \ref{fig:prob_id} provides evidence for this idea in a simulated setting where outcomes are sampled according to \eqref{eqn:N1-model}: as the number of roll-out periods increases, the probability of uniquely identifying the total treatment effect increases very quickly, even in extremely sparse models with an Erdos-Renyi parameter of $0.001$. As expected, the higher the graph density, the likelier we are to satisfy the conditions in Proposition \ref{prop:neighbor_spillovers} and Theorem \ref{thm:uniqueness}. This is because increasing network density also increases the probability that any two units are connected and so generally increases the likelihood that an untreated unit will be connected to a treated unit.

\subsection{Estimation of the Total Treatment Effect}
\label{subsec:est}

In addition to guaranteeing identifiability, temporal variation in the covariates $X$ as measured by the spectrum of $X^\top X$ can also reduce statistical error due to measurement noise. In this section, we quantify how roll-outs provide gains in statistical efficiency under two settings. First, we consider the case of completely correlated noise so that, in each period, outcomes only change as a function of the treatment. This is similar to modeling unobserved individual heterogeneity, which can be fully captured using unit-level fixed effects. Second, and at the other extreme, we consider fully independent noise across individuals and roll-out periods. This setting arises when there are no unobserved temporally persistent events. In both settings we expect roll-outs to improve the final variance bound. 

We begin by studying the bias and variance of our estimator under the assumption that our model selection procedure (cf. Section~\ref{sec:Model Selection}) has correctly chosen an interference model. The unbiasedness of the estimator is clear from the randomized design because $Z$ is drawn randomly in each period independent of $\epsilon$. To study the variance of our estimator, we make the following assumption on how unobserved noise enters the model.

\begin{assumption}[Time-invariant Individual Idiosyncrasies]\label{assum:fixed} 
Suppose $\epsilon_{i,t}=\epsilon_i$ where $\{\epsilon_{i}\}_{i=1,\dots,N,}$ are independent, mean-zero, and satisfies 
$\E[\epsilon_{i}^2] \leq \sigma_{\max}^2 < \infty$.
\end{assumption}

Applying this assumption, we can control the mean squared error ($\MSE$) of our estimator $\hat\TTE$. Since our estimator $\hat\TTE$ is unbiased, we have $\MSE(\hat\TTE) = \E[(\hat\TTE - \TTE)^2] = \var(\hat\TTE)$ and we only have to control the variance term. 
As noise terms across periods are completely correlated under Assumption~\ref{assum:fixed}, idiosyncrasies may persist indefinitely across periods and we expect our variance will increase as a function of $T$. The variance reduction occurs as the population size $N$ grows large, as in the classical linear regression setting. In the below result,  this is captured by the geometry of $X$ using $\lambda_{\min}(X^\top X)$, the minimum eigenvalue of $X^\top X$. 
\begin{theorem}\label{thm:variance_fixed} 
Under the data-generating model~\eqref{eq:additive-model} and Assumption \ref{assum:fixed}, 
we have 
$ \frac{1}{NT} \E\ltwo{X\what{\theta} - X\theta\opt   }^2 \leq 4\sigma_{\max}^2 \frac{K}{N} $, and 
$$
\E[(\hat\TTE - \TTE)^2] 
= \E (c^\top (\what{\theta} - \theta\opt   ))^2 
 \leq 4KT\cdot\ltwo{c}^2 \cdot \sigma_{\max}^2 \cdot \E\left[\lambda_{\min}(X^\top X)^{-1}\right].
 $$ 
\end{theorem}
% \noindent See Section \ref{subsec:pf_thm_var} for the proof.  In it, we bound the prediction error of least squares regression in the setting of perfectly correlated noise, and use it to derive a bound for the variance. This approach allows us to capture the geometry of the minimum eigenvalue of $X^\top X$ through the number of features, $K$. 

It is useful to compare Theorem~\ref{thm:variance_fixed} to the standard linear regression setting with a single period $T=1$, where a similar analysis yields the same bound but without any dependence on $T$. Comparing these two results together, it may seem as though roll-outs increase the variance in the case of time-invariant noise. However, we also need to consider how $\lambda_{\min}(X^\top X)$ scales with $N$ and $T$. By way of illustration, consider the case of a complete graph so that the interference term can be deterministically quantified in relation to $Z^{t}$. In that setting, we find that the minimum eigenvalue grows linearly in $NT$, implying that we recover the classical $\frac{1}{N}$ rate. The following lemma captures this idea that a roll-out allows us to increase our effective data size even with fixed population size and time-invariant errors. In particular, as we have seen in the previous subsection, roll-outs also tend to increase the likelihood that our $X^\top X$ matrix will be full-rank. %Lemma \ref{lemm:conv_complete_linear} showcases this benefit by considering how a better variance bound can be achieved by shrinking the minimum eigenvalue of $X^\top X$ through a roll-out.

\begin{lemma}\label{lemm:conv_complete_linear} 
Consider the same model as in Proposition~\ref{prop:neighbor_spillovers} and a completely randomized roll-out (cf. Definition \ref{def:CRD}) with allocation vector $\vec{p}$. Let Assumption~\ref{assum:fixed} hold and let $f_i$ be  linear-in-means  $f_i(\vec{z})=\frac{1}{|\mathcal{G}(i)|}\sum_{j\in \mathcal{G}(i)}z_j$ where $\mathcal{G}(i)$ is defined by a complete graph. Then, fixing the sample size at $N$ there exists $M\in\R$ large enough that for $NT>M$
$$\E[(\hat\TTE - \TTE)^2] \leq \frac{8\bar{C}_1}{N}, $$
where $\bar{C}_1$ is a constant that depends on $K$, $\sigma^2_{\max}$ from Assumption \ref{assum:fixed}, and the allocation vector $\vec{p}$.
\end{lemma}
\noindent Lemma \ref{lemm:conv_complete_linear} illustrates how roll-outs even when observations are fully correlated across periods still enable us to obtain the usual $\frac{1}{N}$ decay for the variance through $\lambda_{\min}(X^T X)$, the minimum eigenvalue of our Gram matrix. The proof can be found in Section \ref{subsec:pf_lemm_comp}.

We now turn to the case where individual noise is independent across periods. As we have noted earlier, this setting arises when unobserved idiosyncrasies do not persist across several periods. Here we make the following analogue to Assumption \ref{assum:fixed}.
\begin{assumption}[Time-varying Individual Idiosyncrasies]\label{assum:iid}
$\{\epsilon_{i,t}\}_{i\in[N],t\in[T]}$ are independent, mean-zero, and satisfies
$\E\left[\epsilon_{i,t}^2\right] \leq \sigma_{\max}^2 <\infty$.
\end{assumption}
Because Assumption~\ref{assum:iid} requires noise to be independent across time periods, we can achieve tighter control of the variance of our estimator. In particular, roll-outs  decrease the $\MSE$ at a $\frac{1}{T}$ rate since idiosyncrasies are fully independent across time. Applying the same analysis as in our derivation of Theorem \ref{thm:variance_fixed}, we have the following result which we prove in Section \ref{subsec:pf_thm_var}.  
\begin{theorem}\label{thm:variance}
   Under the data-generating model \eqref{eq:additive-model} and Assumption \ref{assum:iid}, we have 
    $ \frac{1}{NT} \E\ltwo{X\what{\theta} - X\theta\opt   }^2 \leq 4\sigma_{\max}^2 \frac{K}{NT} $, and  
    $$\E[(\hat\TTE - \TTE)^2] = \E (c^\top (\what{\theta} - \theta\opt   ))^2 \leq  4K \cdot \ltwo{c}^2 \cdot \sigma_{\max}^2 \cdot \mathbb{E}\left[\lambda_{\min}(X^\top X)^{-1}\right]$$
\end{theorem}

Similar to Lemma \ref{lemm:conv_complete_linear}  in the case of Assumption \ref{assum:iid},  the next lemma provides a concrete bound for a complete interference graph. When the noise is fully independent across periods, we  gain a reduction in variance as $T$ grows. Hence, in this extreme, roll-outs produce an even faster vaster variance reduction through the geometry of $\lambda_{\min}(X^\top X)$.
\begin{lemma}\label{lemm:conv_complete_linear_iid} 
Consider the same model as in Proposition~\ref{prop:neighbor_spillovers} and a completely randomized roll-out (cf. Definition \ref{def:CRD}) with allocation vector $\vec{p}$. Let Assumption~\ref{assum:iid} hold and let $f_i$ be  linear-in-means  $f_i(\vec{z})=\frac{1}{|\mathcal{G}(i)|}\sum_{j\in \mathcal{G}(i)}z_j$ where $\mathcal{G}(i)$ is defined by a complete graph. Then, fixing the sample size at $N$ there exists $M\in\R$ large enough that for $NT>M$
$$\E[(\hat\TTE - \TTE)^2] \leq \frac{8\bar{C}_2}{NT}, $$
where $\bar{C}_2$ is a constant that depends on $K$, $\sigma^2_{\max}$ from Assumption \ref{assum:iid}, and the allocation vector $\vec{p}$.
\end{lemma}
\noindent See Section \ref{subsec:pf_lemm_comp} for the proof.

Figure \ref{fig:variance reduction} shows the implications of Theorems \ref{thm:variance_fixed} and \ref{thm:variance} in both the time-varying and time-invariant noise settings. In both cases, we see a non-trivial variance reduction relative to the no roll-out case, which is emblematic of the variance gains from roll-outs quantified in this section.

\begin{figure}
    \centering
        \subfigure[Time-invariant Individual Idiosyncrasies]{\includegraphics[width=0.45\textwidth]{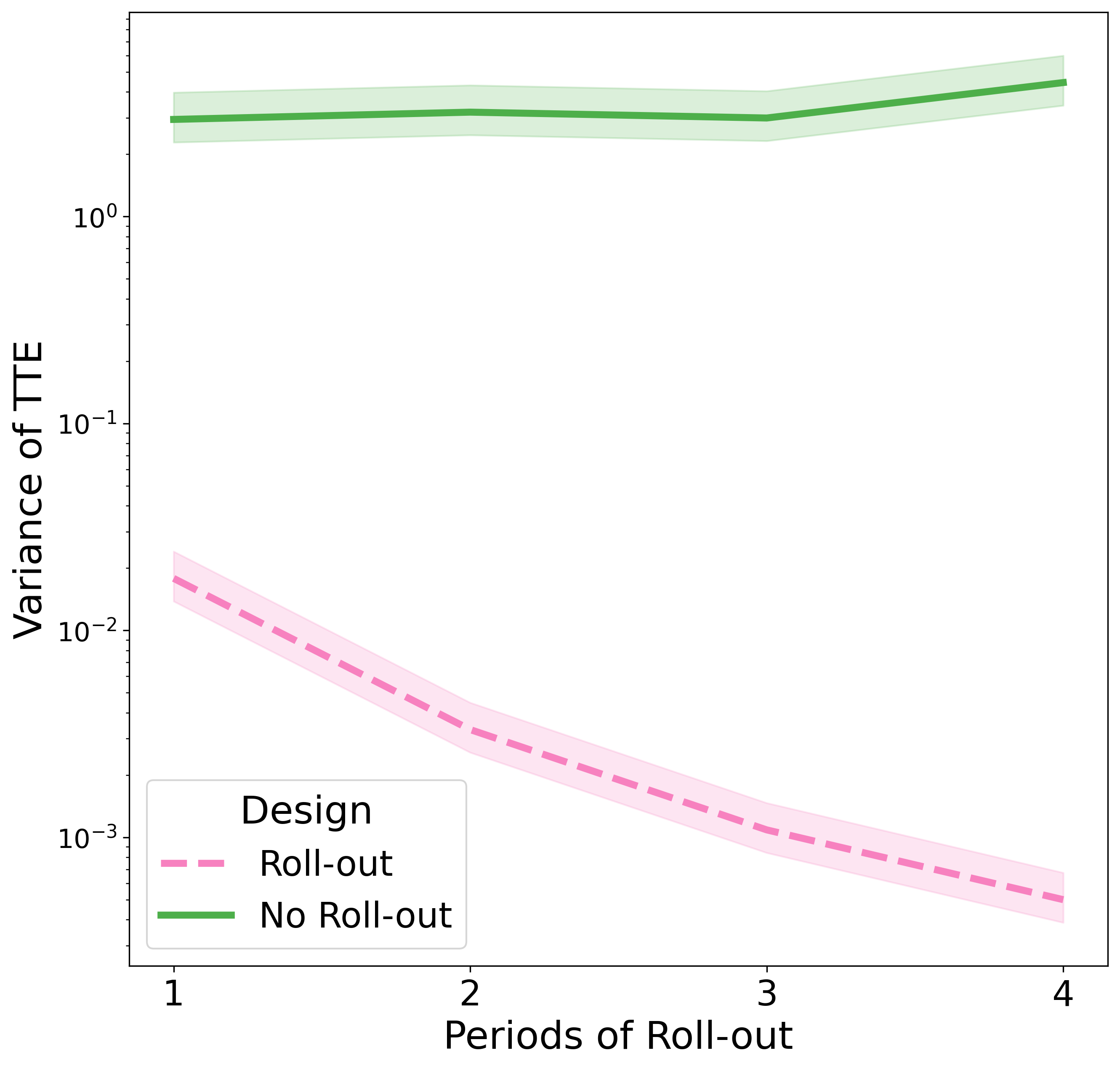}}
        \subfigure[Time-varying Individual Idiosyncrasies]{\includegraphics[width=0.45\textwidth]{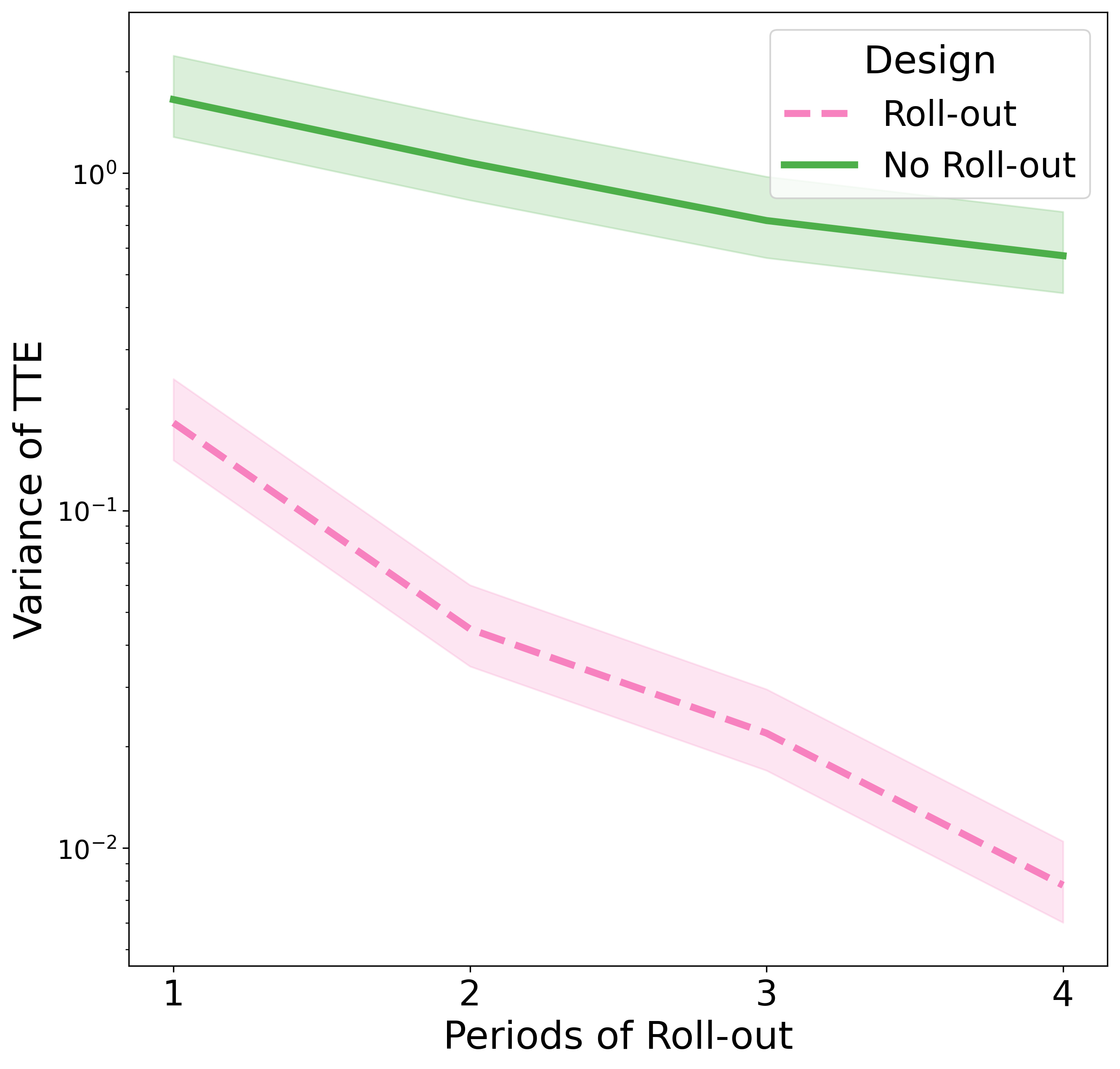}}
    \caption{Variance reduction due to roll-out under fixed and i.i.d. errors. The no roll-out baseline is computed by assuming that we observe a $50\%$ treated population at multiple periods in time with both fixed and random potential outcomes. We display $95\%$ $\chi^2$-confidence intervals.}
\label{fig:variance reduction}
\end{figure}

%In the finite population setting with fixed $N$, roll-outs allow $T$ to grow providing gains in efficiency and implying that our estimator converges to the TTE.

\section{Discussion}
\label{sec:considerations}
In this work, we leverage a universal experimentation design used throughout online platforms, roll-outs, to model interference effects. We quantify how roll-outs induce temporal variation in treatment exposure that facilitates the identification and estimation of the total treatment effect. We propose a model selection procedure to help practitioners model interference and identify the total treatment effect.  We conclude the paper by discussing robustness checks that can augment our model selection framework, and discuss possible pitfalls practitioners may face when applying our methodology. The heuristics we propose below help practitioners implement our methodology. 

\subsection{Practical Considerations}
\label{sec:considerations}

\paragraph{Robustness checks}
A natural question that arises from this analysis is whether there are any robustness checks that can provide evidence that our model selection procedure has chosen the correct model of interference. Fundamentally, testing if a model is correct is not possible. However, there are tests we can perform to build evidence that we are fully accounting for the variation in outcomes caused by interference in our data sample.

The first recommendation we make to practitioners is to include a model without interference terms in the model selection step. Including such a model in our procedure is equivalent to testing for interference. If the model selection procedure chooses the no interference model, when there is a strong prior for interference in the experiment, then this is good evidence that the models that are being tested are inadequately capturing the effects of treatment exposure.

Another possible test uses the interference testing framework of \citet{Li2022}. Their work considers what gains roll-outs provide when attempting to detect interference. They provide several permutation tests under a Bernoulli roll-out design that are able to effectively test for the presence of interference. A key component of these tests is the candidate exposure of each unit which is defined in their notation as $h_i(W_{-i,k})$ where, using our notation $W_{i,k}=Z_{i}^t$. 
%\hn{I found this hard to follow. Is there a typo here?}. 
A simple test to consider is to define $h_i$ to be the interference terms in our setting, that is to say, set $h_i(W_{-i,k}) = \textbf{f}_i(Z^t)$ where $\textbf{f}(\cdot)$ is given by the selected model, and then conduct the proposed multiple experiment test of \cite{Li2022}.  If the test finds interference to be statistically significant, then this is good evidence that the selected model of interference is capturing the effect of treatment exposure on outcomes. While this test may still suffer from misspecification issues, comparing its results to the permutation test proposed by~\citet{Li2022} again can provide strong evidence in favor of the selected model.

A final approach applies the test proposed by \citet{pougetabadie2019testing}. In this case, after model selection is completed, we compute the new outcome, which subtracts the effects of interference from each outcome at every period. Next, we pool our data across all periods and use the underlying interference network to create clusters of units, allowing us to compute a difference-in-means estimate of the total treatment effect and a Horvitz-Thompson estimate under a cluster-based design. We can now compute $\Delta$ as the difference of these estimates and conduct the test proposed in \cite{pougetabadie2019testing}. If we find that the estimates are similar, then this is again evidence that our selected model is accurately capturing interference effects.

\paragraph{Other Considerations}

% Some final considerations for practitioners concern the scalability of our procedure, statistical power, how to define time windows, and auto-regressive terms. 
% In many practical cases, $T$ is small while $N$ is typically very large, so the LOPO procedure is computationally inexpensive. 
% However, the computation can become a bottleneck when $T$ is large.
% Scalability is also a concern when testing for identifiability by checking if  $c\in\text{span}(X^\top)$. However, since roll-outs aid in ensuring $X^\top X$ is non-singular in many cases, this boils down to checking the singularity of the Gram matrix, $X^\top X$. 
% In our setting, this can be done in $O\left((NT)^{2.373}\right)$ time \citep{Williams2014MultiplyingMI}. 

Since effect sizes are typically small in online platforms, the lack of statistical power may result in the inability to distinguish between similar potential outcome models. In many of the simulations we considered, we observed different models yield similar estimates of the TTE, somewhat alleviating such concerns. When considering rich interference models, we recommend a LASSO penalty when conducting estimation. 

There are often non-stationarities in interference effects that require time to equilibrate, and the length of each period in a roll-out is an important design choice. Our procedure relies on the fact that outcomes are observed after the full interference effects have been experienced. Different experiments and settings will naturally require different time windows, and  previous experiments and domain knowledge should guide these choices.  Finally, some practitioners may want to pose auto-regressive models in their experiments. Unfortunately, auto-regressive models pose challenges as they introduce complicated interactions with interference terms, time effects, and individual heterogeneity.  While a well-known practical issue in the context of two-way fixed effects~\citep{ArellanoBond1991DynamicPanel}, the consequences of auto-regressive terms are unclear in terms of interference, which we leave as a topic of future work.
% We observe AR1 terms sometimes lead to substantial variance reductions when paired with interference terms but complicate the identification of the TTE. We believe the variance reductions are largely due to the interference effects having some level of persistence across periods. At this point in our work, we do not consider the benefits of these terms further but leave the question open for future directions.
\subsection{Future Directions}
We summarize several future directions of research. First, we have seen how roll-out schedules can influence our ability to conduct model selection effectively. A theoretical study of model selection requires a formal language for model misspecification in the presence of interference, which we leave for future work. A close study of the design-based perspective posed in our work may yield fruit. While requiring more engineering resources, the experimenter may sometimes be able to adaptively choose a roll-out schedule that maximizes the information that can be learned from the experiment before fully launching an intervention. Similarly, she may mitigate the non-stationarity of interference effects by appropriately choosing periods in a roll-out.  A third direction may consider how to carefully incorporate auto-regressive terms under the presence of interference, which may prove useful from a modeling perspective. Finally, while we have empirically shown that the \textsc{LOPO} procedure tends to perform reliably, we mention some possible extensions in Section~\ref{subsec:limitations}.

\paragraph{Acknowledgement}
We thank Kevin Han, Shuangning Li, Jialiang Mao, and Han Wu for their thoughtful feedback.

\bibliographystyle{apalike}
\bibliography{bibliography} 

\pagebreak
\appendix

\section{Characterizing the Distribution of Roll-out Designs}\label{subsec:rollout}

In this section of the appendix, we analyze two different procedures to implement roll-out designs. The first design, known as the completely randomized roll-out, is defined in Definition \ref{def:CRD} and is used throughout our paper and by contemporaneous work, e.g., \citep{yuetal2022rollout}. In each period $t$, let $\mathcal{S}_t\subseteq[N]$ be the set of newly treated units. The probability that a unit $i$ is selected for treatment  is
%yuetal2022rollout
\begin{equation}
    \P[i \in \mathcal{S}_t] = \frac{Np_t}{N-\lceil N\sum_{j=1}^{t-1}p_j \rceil}.
\end{equation}
Evidently, a completely randomized roll-out  is a Markov chain whose  state transitions occur according to,
$$Z_{i}^t=\begin{cases}1 & Z_{i}^{t-1}=1 \text{ or } i \in \mathcal{S}_t \\ 0 & \text{otherwise}\end{cases}.$$

With this in hand, we can  compute the marginal distribution of $Z_i^t$.
\begin{lemma} \label{lem:crd_dist}
The distribution of $Z^t$ under Definition \ref{def:CRD} is given by
\begin{equation}yuetal2022rollout
 \P[Z_{i}^t=1] = \sum_{k=1}^t p_t
\end{equation}
\end{lemma}
\begin{proof} Let $\mathcal{S}_t\subseteq[N]$ be the set of newly treated units in period $t$.
$$\begin{aligned}
    1-\P[Z_{i}^t=0]&=\P[Z_{i}^t=1]
    \\&=\P[Z_{i}^{t-1}=1] + \P[i \in \mathcal{S}_t]\P[Z_{i}^{t-1}=0]
    \\&=\P[i\in \mathcal{S}_1] + \P[i\in \mathcal{S}_2]\P[Z_{i}^1=0] + \cdots + \P[i\in \mathcal{S}_t]\P[Z_{i}^{t-1}=0]
    \\&= p_1 + \sum_{k=2}^t \frac{Np_k}{N-\lceil N\sum_{j=1}^{k-1}p_j \rceil}\left(1- \P[Z_{i}^{k-1}=1]\right).
\end{aligned}$$
Conclude
$$\begin{aligned}
\P[Z_{i}^1=1]&=p_1
\\ \P[Z_{i}^2=1]&=p_1 + \frac{Np_2}{N(1-p_1)}(1-p_1)=p_1 + p_2
\\&\vdots
\\ \P[Z_{i}^t=1]&=\sum_{k=1}^t p_k.
\end{aligned}$$
\end{proof}
We now turn to the second roll-out implementation that is based on Bernoulli trials. This design appears in \cite{Li2022} and much of the analysis in this paper can also be carried through in this implementation.
\begin{definition}[Bernoulli Roll-outs] \label{def:bern_rollout} A \textbf{$T$-period Bernoulli roll-out} is an increasing set of treatment assignments, $\{Z^1, \dots, Z^T\}$ and treatment proportions $p=\{p_1,\dots,p_T\}$ where $\sum_{t=1}^T p_t = \bar{P} \leq 1$ such that the distribution of $Z$ follows
$$
Z_{i}^{1} \sim \bernoulli(p_1)
~~~\mbox{and}~~~Z_{i}^{t} \sim 
\begin{cases}
    1 & \text{if } Z_{i}^{t-1}=1 \\ 
    \bernoulli(p_t)  & \text{otherwise} 
    \end{cases}
$$
\end{definition}
\begin{lemma}
    \label{lem:dist_bern}
    Suppose that $Z$ is a roll-out given by Definition \ref{def:bern_rollout}. Then, $Z^t$ is a Markov chain with
    \begin{enumerate}
        \item $\P[Z_{i}^{t}=0]=\prod_{j=1}^t (1-p_j)$
        \item $\P[Z_{i}^{t}=1] = p_1 + p_2(1-p_1)+\cdots+p_t\prod^{t-1}_{j=1}(1-p_j)$.
    \end{enumerate}
\end{lemma}
\begin{proof} To see Property 1,  apply the law of total probability and $\P[Z_{i}^t=0|Z_{i}^{t-1}=1]=0$
$$\begin{aligned}
\P[Z_{i}^1=0]&=1-p_1 \\
\P[Z_{i}^2=0]&=\P[Z_{i}^2=0|Z_{i}^1=0]\P[Z_{i}^1=0] = (1-p_2)(1-p_1)\\
&\vdots \\
\P[Z_{i}^{t}=0] &= \prod_{j=1}^t(1-p_t).
\end{aligned}$$
To see Property 2, we again apply the law of total probability
$$\begin{aligned}
\P[Z_{i}^{1}=1]&=p_1 \\
\P[Z_{i}^{2}=1]&= \P[Z_{i}^{2}=1|Z_{i}^{1}=0]\P[Z_{i}^{1}=0] + \P[Z_{i}^{2}=1|Z_{i}^{1}=1]\P[Z_{i}^{1}=1] \\ & = p_2(1-p_1) + 1\cdot p_1 = p_1 + p_2(1-p_1)\\
&\vdots \\
\P[Z_{i}^{t}=1] &= p_1 + p_2(1-p_1)+\cdots+p_t\prod^{t-1}_{j=1}(1-p_j).
\end{aligned}$$
\end{proof}
\section{Proof of Identification Results in Section \ref{sec:Identification}}
\subsection{Proof of Proposition \ref{prop:neighbor_spillovers}}
\label{subsec:pf_prop_spill}
We begin by writing out $X$ 
\begin{align*}
X & =\left[\begin{array}{ccc}
1 & Z_{1}^{1} & f_{1}(Z_{\cdot}^1)\\
\vdots & \vdots & \vdots\\
1 & Z_{N}^{1} & f_{N}(Z_{\cdot}^1)\\
\vdots & \vdots & \vdots\\
1 & Z_{1}^{T} & f_{1}(Z_{\cdot}^T)\\
\vdots & \vdots & \vdots\\
1 & Z_{N}^{T} & f_{N}(Z_{\cdot}^T)
\end{array}\right]\in\mathbb{R}^{TN\times3}.
\end{align*}
$X^\top X$ is non-singular when the columns of $X$ are linearly independent. Suppose by way of contradiction that the columns $X$ are dependent: there is a nonzero vector $\lambda \in \mathbb{R}^3$  such that
$$\lambda_1X_1 + \lambda_2X_2 + \lambda_3X_3 = 0.$$

Recall from our hypothesis that  there is a $i \in [N]$ and $t \in [T]$  with $Z^{t}_i = 0$ and $f_i(Z^{t})\neq 0$. From the linear dependence of the columns, we have 
\begin{equation}
\label{eq:l1_eq_l3}
\lambda_1 + \lambda_3f_i(Z^t) = 0 
\implies \lambda_1 = -\lambda_3f_i(Z^t).
\end{equation}
Now, take $j\in[N]$, $t'\in[T]$ such that $Z_j^{t'}=0$ and $f_i(Z^t) \neq f_j(Z^{t'})$ as assumed in our hypothesis. Linear dependence again implies
\begin{align*}
 \lambda_1 + \lambda_3f_j(Z^{t'}) = 0  
~\iff~ -\lambda_3f_i(Z^t) + \lambda_3f_j(Z^{t'}) = 0  
~\iff~  \lambda_3(f_j(Z^{t'})-f_i(Z^{t})) =0 
~\Rightarrow~ \lambda_3 = 0,
\end{align*}
where we use $f_j(Z^{t'})\neq f_i(Z^{t})$  in the final line. Conclude $\lambda_1 = 0$ from Eq.~\eqref{eq:l1_eq_l3}. 

Now, take any $k\in [N]$ and $q\in[T]$ such that $Z_{k}^q = 1$, which is guaranteed to exist by the definition of the roll-out. Notice that if the columns of $X$ are linearly dependent, $\lambda_1=\lambda_3=0$ implies
\begin{align*}
\lambda_1 + \lambda_2 + \lambda_3f_j(Z^{t'}) &= 0 
\Rightarrow \lambda_2 = 0.
\end{align*}
This is a contradiction since $\lambda \neq \bf{0}$.

\subsection{Proof of Corollary \ref{cor:spillover_identification}}
\label{subsec:pf_cor_id}

Since  $Z_j^t = 1$ and  $i\in \mathcal{G}_1(j) \iff j \in \mathcal{G}_1(i)$, unit $i$'s interference term  in model~\eqref{eqn:N1-model} is positive:
$f_i(\vec{z}) = \sum_{j \in \mathcal{G}_1(i)} \vec{z}_{j} > 0$.
Noting that unit $i$ is untreated $Z_i^t = 0$, 
the hypothesis of Proposition~\ref{prop:neighbor_spillovers} is satisfied for $(i,t)$ and $(i,1)$.

\subsection{Proof of Theorem \ref{thm:uniqueness}}
\label{subsec:pf_thm_id}

To prove this result, we consider two optimization problems that solve for the upper and lower bounds of the total treatment effect.
\begin{equation}\label{eq:opt_procedure}
\begin{aligned}
\max_{\theta\in\mathbb{R}^{k}} & \;c^{\top}\theta & \min_{\theta\in\mathbb{R}^{k}} & \;c^{\top}\theta\\
s.t. & \;X^{\top}X\theta=X^{\top}Y & s.t. & \;X^{\top}X\theta=X^{\top}Y
\end{aligned}
\end{equation}
The dual problems are given by
\begin{equation}\label{eq:opt_procedure_dual}
\begin{aligned}
\min_{p\in\mathbb{R}^{k}} & \;p^{\top}X^{\top}Y & \max_{p\in\mathbb{R}^{k}} & \;p^{\top}X^{\top}Y\\
s.t. & \;p^{\top}X^{\top}X=c^{\top} & s.t. & \;p^{\top}X^{\top}X=c^{\top}
\end{aligned}
\end{equation}

The dual problems are feasible if and only if $c$ lies within
$\text{span}(X^{\top}X) = \text{span}(X^{\top})$ where $\text{span}(\cdot)$ refers to the column span. We show that this condition implies that the lower and upper bounds~\eqref{eq:opt_procedure} are equal,  yielding a unique estimate of the total treatment effect. If $X^{\top}X$ is nonsingular, this is evident. Otherwise, let $$\Theta_0 = \{\theta \in \R^K : X^\top X \theta = X^\top Y \}. $$ Suppose $X^{\top}X$ is singular so that there may exist $\theta,\theta'\in\Theta_{0}$
with $\theta\neq\theta'$. Let $\tilde{\theta},\tilde{p}$ and $\undertilde{\theta},\undertilde{p}$ be the optimal primal-dual pair for the upper and lower bound problems respectively. 
From primal feasibility  $\tilde{\theta},\undertilde{\theta}\in\Theta_{0}$
and dual feasibility $\tilde{p},\undertilde{p}\in\left\{ p:X^{\top}Xp=c\right\}$,
 strong duality gives
\begin{align*}
c^{\top}\tilde{\theta} & =\tilde{p}^{\top}X^{\top}Y
  =\tilde{p}^{\top}X^{\top}X\undertilde{\theta} & \undertilde{\theta}\in\Theta_{0}\\
 & =(X^{\top}X\tilde{p})^{\top}\undertilde{\theta} & \text{\ensuremath{X^{\top}X} is symmetric}\\
 & =c^{\top}\undertilde{\theta} & X^{\top}X\tilde{p}=c\text{ by dual feasibility}
\end{align*}

\subsection{Proof of Theorems \ref{thm:variance_fixed} and \ref{thm:variance}}
\label{subsec:pf_thm_var}

% First, we note that our estimator of the $\TTE$ is unbiased.  $X$ is a random matrix composed of measurable functions of $Z$ which are drawn according to the roll-out designs specified in Lemma \ref{lem:crd_dist} or \ref{lem:dist_bern}. In particular, it is independent of $\epsilon$ so that
% $$\mathbb{E}[c^\top(X^\top X)^{-1}X^\top\epsilon] = \mathbb{E}[\mathbb{E}[c^\top(X^\top X)^{-1}X^\top\epsilon|X]]=\mathbb{E}[c^\top(X^\top X)^{-1}X^\top\mathbb{E}[\epsilon|X]]=0.$$

By the definition of the minimum eigenvalue, we have the tautological bound
$$
\ltwo{\what{\theta} - \theta\opt}^2 
\le  \lambda_{\min}\left(X^\top X\right)^{-1}
\ltwo{X\what{\theta} - X\theta\opt}^2.
$$ 
Conditioning on $X$ so that $\lambda_{\min}(X^\top X)$ is deterministic, Cauchy-Schwarz gives
\begin{equation}
\label{eq:MSE_leq_min_eig}
    \E_X|c^\top(\what\theta - \theta\opt)|^2 
    \le \ltwo{c}^2 
    \E_X\ltwo{\what\theta - \theta\opt}^2
    \le 
    \ltwo{c}^2  
    \lambda_{\min}\left(X^\top X\right)^{-1}
    \cdot \E_X\ltwo{X\what{\theta} - X\theta\opt   }^2.
\end{equation}
To bound the final term in the preceding display, we use an adaptation of \citet[Theorem 2.2]{rigollet2015high}; recall that $X\in \R^{NT\times K}$ and $Y=X\theta\opt + \epsilon$.
\begin{lemma}
\label{lemm:MSE_dim_red} 
Under the linear model of \eqref{eq:additive-model}, suppose either Assumption \ref{assum:fixed} or \ref{assum:iid} hold. 
Then, the least squares estimator $\what{\theta} = (X^\top X)^{-1}X^\top Y$
satisfies
\begin{align}
\label{eq:mspe}
\E_X \ltwo{X\what{\theta} - X\theta\opt   }^2 
    \le
    \begin{cases}
        4\sigma_{\max}^2 KT
        & \mbox{under Assumption~\ref{assum:fixed}} \\
        4\sigma_{\max}^2 K
        & \mbox{under Assumption~\ref{assum:iid}}
    \end{cases}.
\end{align}
\end{lemma}
\noindent Applying the lemma, we have the desired result. 

In the remainder of the section, we prove the bound~\eqref{eq:mspe}. Since $\what{\theta}$ is the least squares estimator, we have
$$\ltwo{Y-X\what{\theta}}^2 \leq \ltwo{Y-X\theta\opt   }^2=\ltwo{\epsilon}^2.$$
As we assume a well-specified linear model $Y = X\theta\opt + \epsilon$, this implies
$$
\ltwo{\epsilon}^2 \ge 
\ltwo{Y-X\what{\theta}}^2= \ltwo{X\theta\opt    + \epsilon -X\what{\theta}}^2 = \ltwo{X\what{\theta} - X\theta\opt   }^2 -2\epsilon^\top X(\what{\theta} - \theta\opt   )+\ltwo{\epsilon}^2.
$$
Rearranging terms, we get
\begin{equation}
\label{eq:cov_bound}
\ltwo{X\what{\theta} - X\theta\opt   } 
\leq 2\frac{\epsilon^\top X(\what{\theta}-\theta\opt   )}{\ltwo{X(\what{\theta}-\theta\opt   )}}.
\end{equation}

Let $\Phi \in \R^{NT\times K}$ be an orthonormal basis for the column span of $X$. Then, there exists $\nu\in\R^K$ such that $X(\what{\theta} -\theta\opt   )=\Phi\nu$. Letting $B_K = \{u\in\R^K:||u||_2\leq 1\}$ be the unit ball in $\R^K$, conclude
$$\begin{aligned}
\frac{\epsilon^\top X(\what{\theta}-\theta\opt   )}
{\ltwo{X(\what{\theta}-\theta\opt   )}} 
= \frac{\epsilon^\top\Phi\nu}{\ltwo{\Phi\nu}} 
& =\frac{\epsilon^\top\Phi\nu}{\ltwo{\nu}} 
& \qquad  \text{by orthonormality of $\Phi$}
\\&\leq \sup_{\mu\in B_k}\epsilon^\top\Phi\mu & \text{since $\nu/||\nu||_2\in B_K$}.
\end{aligned}$$
From Equation~\eqref{eq:cov_bound}, we arrive at
\begin{equation}
\ltwo{X\what{\theta} -X\theta\opt   }^2 
\leq 4 \cdot \left( \sup_{\mu\in B_k}\epsilon^\top\Phi\mu\right)^2.
\end{equation}
Since the last expression is maximized at $\mu = \Phi^\top \epsilon/\ltwo{\Phi^\top \epsilon}$
$$
\E_X\left[\left(
\sup_{\mu\in B_k} \epsilon^\top\Phi\mu\right)^2\right] 
= \E_X\left[\sum_{i=1}^K 
\left([\Phi^\top \epsilon]_i\right)^2\right] 
\leq K \cdot \max_{i=1,\dots,K}\var_X([\Phi^\top \epsilon]_i).
$$

We now consider the case of perfectly correlated and fully independent errors separately.
\paragraph{Case for Assumption \ref{assum:fixed} (Theorem~\ref{thm:variance_fixed})} Recall that under Assumption \ref{assum:fixed}, $\epsilon_{i,t}$ are perfectly correlated across $t$. For convenience, we define the vector $\vec{\sigma}_i = [\sigma_{i,1},\dots,\sigma_{i,T}]^\top\in\R^T$ and let $\Sigma_i = \vec{\sigma_i}\left(\vec{\sigma_i}\right)^\top$. 
We begin by noting that
$$\cov_X(\Phi^\top \epsilon) = \Phi^\top \cov(\epsilon) \Phi  \preceq \lambda_{\max}(\cov(\epsilon)) I_{K\times K},$$
% where the last relation follows from the fact $\lambda_{\max}(\cov(\epsilon))I_{NT\times NT} \succeq \cov(\epsilon)$ since $\cov(\epsilon)$ is symmetric so that $\lambda_{\max}(\cov(\epsilon))I_{NT\times NT} - \cov(\epsilon)$ is positive semi-definite.
where Assumption \ref{assum:fixed} imposes a block-diagonal structure on $\cov(\epsilon)$ such that each block is given by $\Sigma_i$. 
Letting  $\eta = [\eta_1,\dots,\eta_N]^\top\in\R^{NT}$ so that   $\eta_i \in \R^T$,
we bound the variational representation for the maximum eigenvalue of $\cov(\epsilon)$
$$\begin{aligned}\lambda_{\max}(\cov(\epsilon)) &= \max_{\ltwo{\eta} = 1} \eta^\top \cov(\epsilon)\eta
= \max_{\ltwo{\eta} = 1} \sum_{i=1}^{N}\eta_i^\top \Sigma_i \eta_i
% = \max_{\ltwo{\eta} = 1} \sum_{i=1}^{N}\eta_i^\top  \vec{\sigma_i}\vec{\sigma_i}^\top \eta_i
= \max_{\ltwo{\eta} = 1} \sum_{i=1}^{N} \left(\eta_i^\top \vec{\sigma}_i\right)^2
\\& \leq \max_{\ltwo{\eta} = 1} \sum_{i=1}^{N}  
\ltwo{\eta_i}^2 \ltwo{\vec{\sigma}_i}^2
% \\& = \max_{\eta^\top \eta = 1} \sum_{i=1}^{N}  \left( ||\eta_i^\top||_2^2 \cdot \sum_{t=1}^T\sigma^2_{i,t} \right) & \text{where }\sigma^2_{i,t} \defeq \var(\epsilon_{i,t})
%\\& \leq \max_{||\eta||_2 = 1} T\cdot \sigma_{\max}^2 \sum_{i=1}^{N}  ||\eta_i^\top||_2^2
%\\& \leq \max_{||\eta||_2 = 1} T\cdot \sigma_{\max}^2 \sum_{i=1}^{N}  \sum_{t=1}^T \eta_{i,t}
\leq  T \sigma_{\max}^2 \cdot \max_{\ltwo{\eta} = 1} \sum_{i=1}^{N}  
\ltwo{\eta_i}^2 
 = T\cdot \sigma_{\max}^2. 
\end{aligned}$$
Noting that
    $\cov_X(\Phi^\top \epsilon)  \preceq T\cdot \sigma_{\max}^2 \cdot I_{K\times K}$, conclude
    $$\E_X \ltwo{X\what{\theta} - X\theta\opt   }^2 \leq 4\sigma_{\max}^2 KT.$$
\paragraph{Case for Assumption \ref{assum:iid} (Theorem~\ref{thm:variance})} Under Assumption \ref{assum:iid}, $\epsilon_{i,t}$ is independent across both $i$ and $t$. In this case, $\cov(\epsilon)$ is a diagonal matrix with entries $\var(\epsilon_{i,t})$ and evidently, $\lambda_{\max}(\cov(\epsilon))\leq \sigma_{\max}^2$. As before, 
    $\cov_X(\Phi^\top \epsilon)  \preceq  \sigma_{\max}^2 \cdot I_{K\times K}$,
which implies
$$ \E_X \ltwo{X\what{\theta} - X\theta\opt   }^2  \leq 4\sigma_{\max}^2 K.$$
\subsection{Proof of Lemma \ref{lemm:conv_complete_linear} and \ref{lemm:conv_complete_linear_iid}}
\label{subsec:pf_lemm_comp}
%\noindent\textbf{Proof of Lemma \ref{lemm:conv_complete_linear}}
Under the model in Proposition \ref{prop:neighbor_spillovers},  $X^\top X$ can be computed as
$$X^\top X = \left(\begin{array}{ccc}
1 & \cdots & 1\\
Z^{1} & \cdots & Z^{T}\\
f(Z^{1}) & \cdots & f(Z^{T})
\end{array}\right)\cdot\left(\begin{array}{ccc}
1 & Z^{1} & f(Z^{1})\\
\vdots & \vdots\\
1 & Z^{T} & f(Z^{T})
\end{array}\right)
	=\left(\begin{array}{ccc}
NT & \sum_{i,t}Z_{i}^{t} & \sum_{i,t}f_{i}(Z^{t})\\
\sum_{i,t}Z_{i}^t & \sum_{i,t}Z_{i}^t & \sum_{i,t}Z_{i}^tf_{i}(Z^{t})\\
\sum_{i,t}f_{i}(Z^{t}) & \sum_{i,t}Z_{i}^tf_{i}(Z^{t}) & \sum_{i,t}\left(f_{i}(Z^{t})\right)^{2}
\end{array}\right)$$
where $Z^t = [Z^t_1,\dots,Z^t_N]^\top$ and $f(Z^t) = [f_1(Z^t),\dots,f_N(Z^t)]^\top$.  Under a completely randomized design with allocation vector $\vec{p}$ and linear-in-means interference, we can fully characterize the design matrix $X$. Specifically, we know that the first column is the vector ${\bf{1}}\in\R^{NT}$ (i.e. the intercept term), and the second column is the stacked treatment vectors $[Z^1 ,\dots,Z^T]^\top\in\R^{NT}$. The last column is the interference term which can be computed exactly because we know since $f_i(Z^t) = \frac{1}{|\mathcal{G}(i)|}\sum_{j\in \mathcal{G}(i)}Z_j^t$ and $\mathcal{G}$ is defined by a fully connected network. Therefore, $|\mathcal{G}(i)|=N-1$ for every $i\in[N]$. Then if unit $i$ is treated ($Z_i^t = 1$) at period $t$ we have
$$f_i(Z^t) = \frac{(\sum_{k=1}^t p_k)N-1}{N-1}.$$
Otherwise, if unit $i$ is untreated ($Z_i^t = 0$) at period $t$ we have
$$f_i(Z^t) = \frac{(\sum_{k=1}^t p_k)N}{N-1}.$$
From here we can easily compute the elements of the matrix $X^\top X$ which are given by
$$\begin{aligned} \sum_{i,t} Z_{i}^t &= \sum_{i,t} f_i(Z^t) = N(Tp_1 + (T-1)p_2 + \cdots + 1\cdot p_T) = N\sum_{t=0}^{T-1}t\cdot p_{T-t},
\\ \sum_{i,t} Z_{i}^tf_i (Z^t) &= \frac{p_1N-1}{N-1}p_1N + \cdots + \frac{(N\sum_{t=1}^{T-1}p_t -1)}{N-1}\left(N\sum_{t=1}^{T-1}p_t\right) + \frac{(N\sum_{t=1}^Tp_t -1)}{N-1}\left(N\sum_{t=1}^Tp_t\right)
\\&= \sum^T_{J=1}\left(\frac{\left(\sum_{t=1}^Jp_t\right)N-1}{(N-1)}\right)^2\left(N\sum_{t=1}^Jp_t\right), \text{ and}
\\ \sum_{i,t} f_i(Z^t)^2 &= \sum^T_{J=1}\left(\frac{\left(\sum_{t=1}^Jp_t\right)N-1}{(N-1)}\right)^2\left(N\sum_{t=1}^Jp_t\right) + \left(\frac{\left(\sum_{t=1}^Jp_t\right)N}{(N-1)}\right)^2\left(N\left(1-\sum_{t=1}^Jp_t\right)\right).\end{aligned}$$

Plugging these into our matrix $X^\top X$, we can now solve the characteristic polynomial, $\det(\lambda I - X^\top X) = 0$ as $N\to\infty \text{ and } T\to \infty$ and examine how the eigenvalues scale with $NT$. Due to the analytic complexity of the problem, we compute this in Mathematica using the \textsc{AsymptoticSolve} method which yields that for $NT$ large enough
$$\lambda_i \asymp \frac{NT}{C_i(\vec{p})} \text{ for all $i$}$$
where $C_i(\vec{p})$ is a function that only depends on the increment vector $\vec{p}$.

\paragraph{Case for Assumption \ref{assum:fixed} (Lemma \ref{lemm:conv_complete_linear})} Define $\bar{C}_1 = K \cdot \sigma_{\max}^2 \cdot C_{i\opt   }(\vec{p})$ where $i\opt    = \arg\min_i \lambda_i$ and $K$ is fixed under the model. Under this model $c=[0,1,1] \implies \ltwo{c}^2 = 2$. Applying Theorem \ref{thm:variance_fixed} and plugging in these values yields the desired result: for some $M\in \R$ and $NT>M$
$$\E\left[|c^\top(\theta\opt   -\what{\theta})|^2\right] \leq KT \cdot \sigma^2_{\max} \cdot \left(\frac{NT}{C_i(\vec{p})}\right)^{-1} = \frac{8\bar{C}_1}{N}.$$
\paragraph{Case for Assumption \ref{assum:iid} (Lemma \ref{lemm:conv_complete_linear_iid})}
 Define $\bar{C}_2 = K \cdot \sigma_{\max}^2 \cdot C_{i\opt   }(\vec{p})$ where $i\opt    = \arg\min_i \lambda_i$ and $\sigma_{\max}^2$ is from Assumption \ref{assum:iid}.  
 Again, $c=[0,1,1] \implies \ltwo{c}^2 = 2$. Plugging in these values and noting that $T$ when using Theorem \ref{thm:variance}.
 Conclude that for some $M\in\R$ large enough, whenever $NT>M$
$$\E\left[|c^\top(\theta\opt   -\what{\theta})|^2\right] \leq K \cdot \sigma^2_{\max} \cdot \left(\frac{NT}{C_i(\vec{p})}\right)^{-1} = \frac{8\bar{C}_2}{NT}.$$

\section{Estimation of the TTE under Model Misspecification}
In Figure \ref{fig:yucompare}, we assess the performance of our estimator based on the potential outcomes model specified in \cite{yuetal2022rollout} (Sec. 5 Eq. 4) and reproduced in \eqref{eq:PI_simulation}. We apply $\textsc{LOPO}$ to variations of the model in \ref{eq:true_model} including a model with a second-order term. In the following experiment, we generate data using the model~\eqref{eq:PI_simulation} with $\beta=2$. We also consider a misspecified Lagrange interpolation estimator of \cite{yuetal2022rollout} with $\beta \neq 2$. To facilitate comparison, we follow \cite{yuetal2022rollout} and do not include a noise term so that any error is due to model misspecifcation alone.

\label{sec:misspec}
\begin{figure}[ht]
    \centering
        \includegraphics[width=0.7\textwidth]{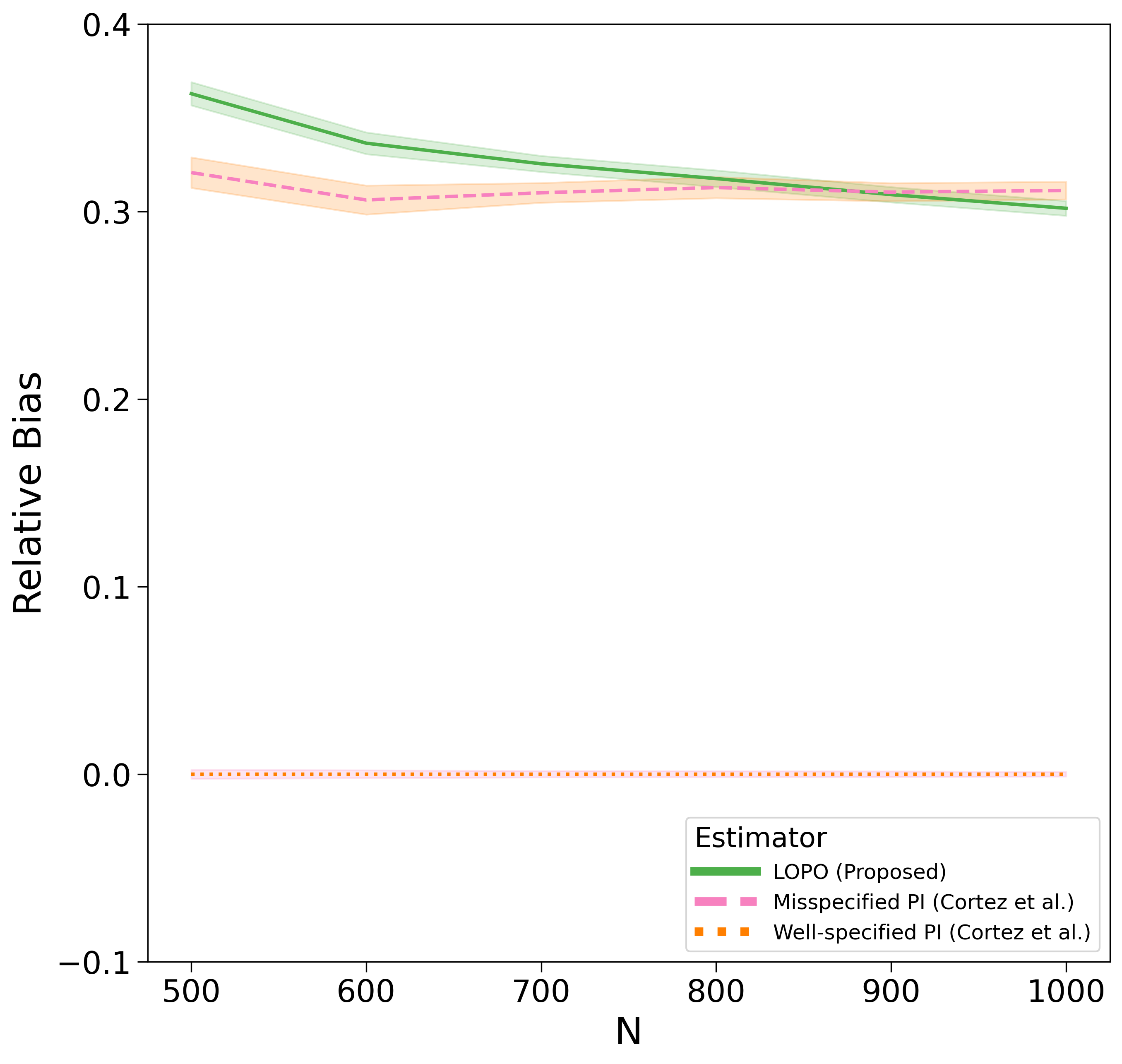}
    \caption{Relative Bias of Estimators of the TTE: figures are generated by averaging the results of 100 experiments. In each experiment sample size is given by the x-axis, $N$, with $T=4$ and an uneven roll-out. We display bootstrapped $95\%$ confidence intervals of the relative bias. We use the DGP given by \cite{yuetal2022rollout} in Eq. \eqref{eq:PI_simulation} with $r=2$ and $\beta=2$. Note that the correctly specified PI estimator (with $\beta =2$) has constant zero bias, hence there is no visible CI.}
\label{fig:yucompare}
\end{figure}

Firstly, we validate that in the well-specified case, we reproduce the results of \cite{yu2022estimating}, which demonstrates that their estimator has zero bias. As the sample size increases,  our misspecified estimator using \textsc{LOPO} preforms similarly to the Lagrange interpolation estimator. The phenomenon shows how linear models are able to approximate polynomial models. This also underscores the need to consider model misspecification in practice as even slight changes in the interference model estimated can result in possibly large relative biases.

\end{document}